\documentclass[a4paper,UKenglish]{lipics}
\usepackage{microtype}%if unwanted, comment out or use option "draft"
\usepackage{amsmath,amssymb,graphicx}
\usepackage{algorithm}
\usepackage[noend]{algorithmic}
\usepackage{snapshot}

\newcommand{\IR}{\ensuremath{\mathbb{R}}}

\newcommand{\lee}{\leqslant}
\newcommand{\gee}{\geqslant}

\newcommand{\eps}{\varepsilon}

\newcommand{\REM}[1]{}

\newcommand{\eq}{{\ \leftarrow\ }}

\newcommand{\CR}{{\mathscr R}}
\newcommand{\CB}{{\mathscr B}}
\newcommand{\CS}{{\mathscr S}}
\newcommand{\CC}{{\mathscr C}}
\newcommand{\CO}{{\mathscr O}}
\newcommand{\Pol}{{\mathscr P}}

\newcommand{\PolSeg}{{P}}

\newcommand{\CQ}{{\mathscr Q}}

\newcommand{\CH}{{\mathscr CH}}

\newcommand{\gre}{{g}}
\newcommand{\sma}{{s}}
\newcommand{\cfev}{{l}} %curve for each variable 
\newcommand{\sq}{{\CS\CQ}}

\newcommand{\R}{\CR}

\newcommand{\Frechet}{Fr\'echet }
\newcommand{\distF}{\delta_F}
\newcommand{\DoubleB}{Double-TypeB}

\DeclareMathOperator{\Left}{left}
\DeclareMathOperator{\Right}{right}

\newcommand{\Seg}[1]{{\overline{#1}}}

\newcommand{\Dir}{\overrightarrow}

\newcommand{\lei}{\prec}
\newcommand{\lex}{\preceq}

\newcommand{\ap}{\oplus}
\newcommand{\pset}{S}
\newcommand{\Qs}{x} % the point the we choose in S to run cccg11 alg
\newcommand{\Ps}{x'} %a point in eps distance to z on the boundary
\newcommand{\Good}{\mbox{type A}}
\newcommand{\SemiBad}{\mbox{type B}}
\newcommand{\Bad}{\mbox{type C}}

\newcommand{\see}{\leadsto}

\newcommand{\lme}[1]{{\lambda_{#1}}}
\title{Visiting All Sites with Your Dog 
\footnote{Submitted to STACS 2013, on Sep 21, 2012}
\footnote{Research supported by NSERC}}
\author[1]{Anil Maheshwari}
\author[2]{J\"{o}rg-R\"{u}diger Sack }
\author[3]{Kaveh Shahbaz}

\affil[1,2,3]{~School of Computer Scince \\ Carleton University 
\\ Ottawa, Ontario, Canada \\
Email: {\tt \{anil,sack,kshahbaz\}@scs.carleton.ca}}

% mandatory: Please provide 1-5 keywords
\serieslogo{}%please provide filename (without suffix)
\volumeinfo%(easychair interface)
  {Billy Editor, Bill Editors}% editors
  {2}% number of editors: 1, 2, ....
  {Conference title on which this volume is based on}% event
  {1}% volume
  {1}% issue
  {1}% starting page number
\EventShortName{}
\DOI{10.4230/LIPIcs.xxx.yyy.p}% to be completed by the volume editor
\begin{document}

\maketitle

\begin{abstract}
Given a polygonal curve $P$, a pointset $\pset$,
and an $\eps > 0$,
we study the problem of finding a polygonal curve $Q$
whose vertices are from $\pset$ and 
has a \Frechet distance less or equal to $\eps$ to curve $P$.
In this problem, $Q$ must visit every point in $\pset$
and we are allowed to reuse points of pointset in building $Q$.
First, we show that this problem in NP-Complete. 
Then, we present a polynomial time algorithm for 
a special cases of this problem, when 
$P$ is a convex polygon.
\end{abstract}

\section {Introduction}

Geometric pattern matching and recognition %consisting of line segments 
has many applications in 
geographic information systems, computer aided design,
molecular biology, computer vision, traffic control, medical imaging etc.
Usually these patterns consist of line segments and polygonal curves. 
\Frechet metric is one of the most popular ways to measure the similarity of two curves. 
An intuitive way to illustrate the \Frechet distance is as follows.
Imagine a person walking his/her dog, where the person and the dog, 
each travels a pre-specified curve, from beginning to the end, 
without ever letting go off the leash or backtracking.
The \Frechet distance between the two curves is the minimal length of a leash which is necessary.
The leash length determines how similar the two curves are to each other:
a short leash means the curves are similar,
and a long leash means that the curves are different from each other.

Two problem instances naturally arise:  decision and optimization.
In the {\em decision problem}, one wants to decide whether two polygonal curves $P$  and $Q$
are within $\eps$ \Frechet distance to each other, i.e., if a leash of given length $\eps$ suffices.
In the {\em optimization problem}, one wants to determine the minimum such $\eps$.
In~\cite{AltG95}, Alt and Godau gave an $O(n^2)$ algorithm for the decision problem,
where $n$ is the total number of segments in the curves.
They also solved the corresponding optimization problem in $O(n^2\log n)$ time.

In this paper, we address the following variant of the \Frechet distance problem.
Consider a point set $S \subseteq \IR^d$ and a polygonal curve $P$ in $\IR^d$, 
for $d \gee 2$ being a fixed dimension.
The objective is to decide whether there exists a polygonal curve $Q$ within an  $\eps$-\Frechet
distance to $P$ such that the vertices of $Q$ are all chosen from the 
pointset $\pset$. Curve $Q$ has to visit every point of $\pset$
and it can reuse points. We show that this problem is NP-Complete. 
We then present a polynomial time decision algorithm for a 
special case of the problem where the input curve $P$ is a convex polygon. 

This paper is organized as follows: in Section \ref{sec:NPComp}, 
we establish our NP-Complete proof for the general case of the problem. 
In Section \ref{sec:SpecialCase}, 
we investigate the special case of the problem.
Finally, we conclude in Section \ref{sec:conc} with some open problems.

\section{General Case is NP-Complete}
\label{sec:NPComp}

\subsection{Preliminaries}
Given two curves $\alpha, \beta: [0,1] \rightarrow \IR^d$,
the {\em \Frechet distance\/} between $\alpha$ and $\beta$ is defined as
$
	\distF(\alpha,\beta) = \inf_{\sigma, \tau} \max_{t \in [0,1]} \| \alpha(\sigma(t)), \beta(\tau(t)) \|,
$
where $\sigma$ and $\tau$ range over all strictly monotone increasing continuous functions.
The following two observations are immediate.

\newtheorem{obs}{Observation}

%\begin{obs}\label{obs:simple}
\begin{obs}\label{obs:simple}
	Given four points $a, b,c,d \in \IR^d$, if
	$\| ab\| \lee \eps$ and $\| cd\| \lee \eps$, then
	$\distF(\Dir{ac},\Dir{bd}) \lee \eps$. 
\end{obs}

\begin{obs}\label{obs:concat}
%\begin{lemma}\label{obs:concat}
	Let $\alpha_1$, $\alpha_2$, $\beta_1$, and $\beta_2$ 
	be four curves 
	such that $\distF(\alpha_1,\beta_1) \lee \eps$ and
	$\distF(\alpha_2,\beta_2) \lee \eps$. 
	If the ending point of $\alpha_1$ (resp., $\beta_1$), 
	is the same as 
	the starting point of $\alpha_2$  (resp., $\beta_2$),
	then $\distF(\alpha_1 + \alpha_2, \beta_1 + \beta_2) \lee \eps$,
	where $+$ denotes the concatenation of two curves.
\end{obs}

\vspace{0.2 in}
\hspace{-0.2 in}{\bf Notations.}
%Given a polygonal curve $\ell$ and a point $p$ in the plane,
%by $\app(\ell, p)$, we mean appending point $p$ to $\ell$
%by connecting the endpoint of $\ell$ to $p$.
%
We denote by $P = <p_1p_2p_3...p_n>$, a polygonal curve $P$
with vertices $p_1 p_2 \dots p_n$ in order 
and by $start(P)$ and $end(P)$, we denote 
the starting and ending point of $P$, respectively.
For a curve $P$ and a point $x$, by $P \ap x$, 
we mean connecting $end(P)$ to point $x$
(we use the same notation $P \ap Q$ to show the concatenation of 
two curves $P$ and $Q$).
Let $M(\Seg{ab})$ denote the  midpoint of the line segment $\Seg{ab}$. 
For a point $q$ in the plane, let $x(q)$ and $y(q)$
denote the $x$ and $y$ coordinate of $q$, respectively.

For two line segments $\Seg{ab}$ and $\Seg{cd}$, with $ \Seg{ab} \dashv	\Seg{cd}$,
we denote the intersection point of them. 
Also, 
for a point $a$ and a line segment $\Seg{bc}$,
$a \perp \Seg{bc}$ denotes the point on $\Seg{bc}$ located 
on the perpendicular from $a$ to $\Seg{bc}$. Also, 
$dist(a, \Seg{bc})$ denotes the distance 
between $a$ and segment $\Seg{bc}$.

\begin{definition} \label{def:feasible}
Given a pointset $\pset$ in the plane, let $Curves(\pset)$
be a set of polygonal curves $Q = <q_1 q_2 \dots q_n> $ where: 
$$  \forall{q_i} : q_i \in \pset \mbox{  and  } $$ 
$$  \forall{a} \in \pset: \exists{q_i} \mbox{  s.t. }  q_i = a \mbox{    } $$
\end{definition}

\begin{definition} \label{def:feasibleNPC}
Given a pointset $\pset$, a polygonal curve $P$ and a distance $\eps$, 
a polygonal curve $Q$ is called {\em feasible} if: 
$Q \in Curves(S)$ and $  \distF(P,Q) \le \eps$.
\end{definition}

We show that the problem of deciding whether a 
feasible curve exists or not  is NP-complete.
It is easy to see that this problem is in NP, since 
one can polynomially check whether $Q \in Curves(S)$
and also $\distF(P,Q) \le \eps$, using the algorithm in \cite{AltG95}.

%we can give a certificate which can check in 
%whether a feasible curve $Q$ is in $\eps$-\Frechet distance to 
%curve $P$ of size $n$. %We reduce 3SAT to our problem.
%

\subsection{Reduction Algorithm}

We reduce in Algorithm~\ref{alg:reduction},
an instance of 3CNF-SAT formula $\phi$ 
to an instance of our problem.
The input is a boolean formula 
$\phi$ with $k$ clauses $C_1, C_2, \dots, C_k$ and $n$ variables $x_1, x_2, \dots,x_n$ 
and the output is a pointset $\pset$,
a polygonal curve $P$ in the plane and 
a distance $\eps = 1$.
%The pointset is constructed through lines \ref{l:makeSStart}
%to \ref{l:makeSEnd} of that algorithm and the polygonal curve 
%is built in line \ref{l:makeP} to the end. 

We construct the pointset $S$ as follows.
For each clause $C_i$, $1 \le i \le k$, in the formula $\phi$, 
we place three points $\{s_i, g_i,c_i\}$, refereed by $cl_i$ points, 
in the plane, which are computed
in the $i$-th iteration of Algorithm \ref{alg:reduction} (from line \ref{l:makeSLoop} to line \ref{l:EndLoopPointSet}).
We define $o_i$ to be $M(\Seg{\sma_i\gre_i})$.
By $\sq_i$, $1\le i \le k$, we denote
a square in the plane, centered at $o_i$, 
with diagonal $\Seg{\sma_i \gre_i}$. 
We refer to $\sq_i$, $1\le i \le k$, as
{\em c-squares}. 
For an example of a pointset $S$ corresponding to a formula, 
see Figure \ref{fig:pathAExample}.

Our reduction algorithm constructs the polygonal curve $P$ 
% calling method $ConstructP$ in 
%Algorithm \ref{alg:MakeP}. 
%The algorithm builds $P$
 through $n$ iterations. In  the $i$-th iteration, $1 \le i \le n$,
it builds a subcurve $\cfev_i$ corresponding to  a variable 
$x_i$ in the formula $\phi$ and appends that curve to $P$.
In addition to those $n$ subcurves, two curves 
$\cfev_{n+1}$ and $\cfev_{n+2}$ are appended to $P$. 
We will later discus the reason we add those 
two curves to $P$. 
Every subcurve  $\cfev_i$ of $P$ starts at point $u$ 
and ends at point $v$.
Furthermore, every $\cfev_i$  goes through c-squares
$\sq_1$ to $\sq_k$ in order, enters each  $\sq_j$ from
the side $\Seg{c_j s_j}$  and exists that square from 
the side $\Seg{c_j g_j}$  (for an illustration, see Figure \ref{fig:pathAExample}). 
Curve $\cfev_i$ itself is built incrementally  
through iterations of the loop at line \ref{l:looptoMakeL} 
of Algorithm \ref{alg:reduction}. 
In the $j$-th iteration, when $\cfev_i$ goes through $\sq_j$,
three points, which are within  $\sq_j$, are added to  $\cfev_i$ 
(these three points are computed through lines 
\ref{l:makeclausestart} to \ref{l:makeclauseend}).
%when the curve goes through square $\sq_j$: 
Next, before $\cfev_i$ reaches to $\sq_{j+1}$,
two points,  denoted by $\alpha_j$ and $\beta_j$, are added to that curve 
(these two points are computed in lines \ref{l:alpha} 
and \ref{l:beta}).
%The polygonal curve $P$ and pointset $S$, which are constructed by 
%our reduction algorithm has the following property:

%Recall that subcurve $\cfev_i$ in $P$, corresponds to variable $x_i$ in $\phi$. 

Since each $\cfev_i$ corresponds to variable $x_i$ in our approach, 
this is how we simulate $1$ or $0$ values of $x_i$:
Consider a point object $\CO_L$ 
traversing $\cfev_i$, from starting point $u$ to ending point $v$. 
Consider 
another point object $\CO_2$ which wants to 
walk from $u$ to $v$
on a path whose vertices are from points in $S$ and it wants to stay in distance one 
to $\CO_L$. We will show that 
by our construction, object $\CO_2$ has two options, either taking 
the path $A = <u,s_1,g_2,s_3 \dots v>$ or the path $B = <u,g_1,s_2,g_3 \dots v>$ 
(See Figure \ref{fig:pathAExample} and \ref{fig:pathBExample} for an illustration). 
Choosing path $A$ by $\CO_2$ means $x_i = 1$ and choosing path $B$ means $x_i = 0$.
We first prove in Lemma \ref{lemma:PathA} that $\distF(\cfev_i,A) \le 1 $ and 
in Lemma  \ref{lemma:PathB} that $\distF(\cfev_i,B) \le 1 $.
Furthermore, by
Lemma \ref{lemma:NoSwitchFromAtoB}, we prove that 
that as soon as $\CO_2$ chooses 
the path $A$ at point $u$ to walk towards $v$, 
it can not switch to any vertex on path $B$.
In addition, in lemmas \ref{lemma:ABCanSeeC} and \ref{lemma:NOTABCanSeeC}, we prove that 
if $x_i$ appears in the clause $C_j$,
$\CO_2$ could visit point $c_j$ via the path $A$ and not $B$. In contrast, 
when $\neg x_i$ appears in the clause $C_j$,
$\CO_2$ could visit point $c_j$ via the path $B$ and not $A$.
However, when both of  $x_i$ and $\neg x_i$ does not appear 
in $C_j$, $\CO_2$ can not take $A$
or $B$ to visit $c_j$.

%\vspace{0.5em}
\begin{algorithm} 
\caption {{\sc Reduction Algorithm}} 
\label{alg:reduction}
\algsetup{indent=1.5em}
\begin{algorithmic}[1]
	\baselineskip=1.\baselineskip
	\REQUIRE  3SAT formula $\phi$ with $k$ clauses $C_1 \dots C_k$ and $n$ variables $x_1 \dots x_n$

	\vspace{0.1in}
		
	\hspace{-0.2in} {\bf Construct pointset $S$:}  

	\STATE $\pset \leftarrow \emptyset$ \label{l:init}

%$S = S \cup \{u\}$ \label{l:makeSEnd}

	\STATE $\gre_1 = (1,1) $ \label{l:makeSStart}

	\FOR {$j = 1$ to $k$}   \label{l:makeSLoop}

	\STATE $\sma_i \leftarrow \big(x(\gre_j)-2,y(\gre_j)-2\big) $
		 \STATE $o_j   \eq  M(\Seg{\sma_j\gre_j})$

		\IF {($j$ is odd)	}

	\STATE $c_j \leftarrow \big(x(s_j), y(\gre_j) \big)$, $w_j  \eq \big(x(o_j)+\frac{1}{4},  y(o_j)-\frac{1}{4} \big)$

	\STATE  $\gre_{j+1} \eq \big (   x(s_{j}) + \frac{1}{4} + 8, y(s_{j}) + \frac{7}{4} +15 \big)$     		\label{l:ComputeNextEven}
    
		\ELSE
	\STATE $c_j \leftarrow \big(x(g_j), y(s_j) \big)$, $w_j  \eq \big(x(o_j)-\frac{1}{4},  y(o_j)+\frac{1}{4} \big)$ 
    
	\STATE  $\gre_{j+1} \eq \big ( x(s_{j}) + \frac{7}{4} + 15, y(s_{j}) + \frac{1}{4} + 8 \big)$     		\label{l:ComputeNextOdd}

%\STATE   $z_i   \eq M(\Seg{w_ic_i})$

	\ENDIF

		%	\STATE $\gre_{i+1} \eq Enp ( \Dir{s_i M(c_iw_i)}, 10$ )  		\label{l:ComputeNext}

\STATE $z_j = M(\Seg{c_jw_j})$

    \STATE $S = S \cup \{\sma_j,\gre_j, c_j\}$   \label{l:EndLoopPointSet}

	\ENDFOR

	\IF {($k$ is odd)}  \label{l:ComputeV}
	  \STATE $\eta \eq \big( x(o_k)+1, y(o_k)+4\big)$  
      \STATE $v \eq \big( x(o_k)+1, y(o_k)+9\big)$  
		\ELSE
	   \STATE $\eta \eq \big( x(o_k)+4, y(o_k)+1\big)$  
      \STATE $v \eq \big( x(o_k)+9, y(o_k)+1\big)$  
	\ENDIF

	\STATE $t \eq \big(x(v),y(u) -20\big)$
	\STATE $u = (-9,-1)$

    \STATE $\pset = \pset \cup \{u, v, t\}$ \label{l:makeSEnd}

\vspace{0.15in}

	\hspace{-0.25in} {\bf Construct polygonal curve $P$:}

	\STATE $P \eq \emptyset$ \label{l:makeP}

	\STATE $P \eq P \ap t$

	\FOR { $i = 1$ to $n+2$  }   \label{l:mainstart}

		\STATE  $\cfev_i \eq \emptyset$ \label{l:startofL}

		\STATE  $\cfev_{i} \eq \cfev_{i} \ap u$ %$\cfev_{i}^0 \eq u$ 
		\STATE $\cfev_{i} \eq \cfev_{i} \ap (-4,-1) $  \label{l:Adduh1toell}
		\FOR {$j = 1$ to $k$} \label{l:looptoMakeL}

			\IF { ($x_i \in C_j$ and $j$ is odd ) or ($\neg x_i \in C_j$ and $j$ is even ) } \label{l:makeclausestart}	
			%\STATE $\cfev_i$.add$( sort_{xy}$( $M(\Seg{a_jc_j} ),c_j,w_j) )$
			\STATE $\cfev_{i} \eq \cfev_{i} \ap M(\Seg{s_jc_j}) \ap c_j \ap w_j  $
			\ELSIF {{ ($\neg x_i \in C_j$ and $j$ is odd ) or ($ x_i \in C_j$ and $j$ is even ) }}
			%\STATE $\cfev_i$.add$( sort_{yx}$( $M(\Seg{b_jc_j} ),c_j,w_j) )$
			\STATE $\cfev_{i} \eq \cfev_{i} \ap w_j \ap c_j \ap M(\Seg{g_jc_j} ) $
			\ELSE		
			\STATE $\cfev_i \eq \cfev_i w_j \ap c_j \ap w_j$
			\ENDIF \label{l:makeclauseend}

			\IF {$j \neq k$}

			\STATE $\alpha_j =  \frac{4}{5} g_j + \frac{1}{5} g_{j+1}$ \label{l:alpha}

			\STATE $\beta_j = \frac{1}{5} s_j + \frac{4}{5} s_{j+1} $ \label{l:beta}
			
			\STATE $\cfev_i \eq \cfev_i \ap \alpha_j \ap \beta_j$

			\ENDIF

		\ENDFOR
	
		\STATE $\cfev_i \eq \cfev_i \ap \eta \ap v$ \label{l:subcurve}

		\STATE  $P \leftarrow P \ap \cfev_i$ 
		\STATE  $P \eq P \ap t$

	\ENDFOR

	\vspace{0.05in}
\RETURN  pointset $\pset$, polygonal curve $P$ and distance $\eps = 1$

\end{algorithmic}
\end{algorithm}

\begin{lemma}\label{lemma:PathA}
Consider any subcurve $\cfev_i =<u \dots v>$, $1\le i \le n+2$,  
which is built through lines \ref{l:mainstart} to \ref{l:subcurve} 
of Algorithm \ref{alg:reduction}. Let $A$ be the polygonal curve  $<u\sma_1\gre_2\sma_3\gre_4..v>$. Then, $\distF(\cfev_i,A) \le 1$.
\end{lemma}

\begin{proof}

We prove the lemma by induction on the number of segments along $A$. 
Consider two point objects $\CO_L$ and $\CO_A$ 
traversing $\cfev_i$ and $A$, respectively (Figure \ref{fig:pathAExample} depicts an instance of $\cfev_i$ and $A$).
We show that $\CO_L$ and $\CO_A$ can walk
their respective curve, from the beginning to
 end, while keeping distance $1$ to each other. 

The base case of induction trivially holds as follows 
(see Figure \ref{fig:PathAClause1} for an illustration):
Table \ref{tab:BaseCasePathA} lists  pairwise location of 
$\CO_L$ and $\CO_A$, where the distance of each pair is at most $1$.
Hence, $\CO_A$ can walk from $u$ to $s_1$ on the 
first segment of $A$ (segment $\Dir{us_1}$), 
while keeping distance $\le 1$ to $\CO_L$.

Assume inductively that $\CO_L$ and $\CO_A$ have feasibly walked along 
their respective curves, until $\CO_A$ reached $s_j$.
Then, as the induction step, 
we 
show that
$\CO_A$ can walk to $g_{j+1}$ and then to $s_{j+2}$ 
%along $\Dir{s_jg_{j+1}}$ 
, while keeping distance $1$ to $\CO_L$.
Table \ref{tab:PathA} lists pairwise location 
of $\CO_A$ and $\CO_L$ such that $\CO_A$ could reach  $s_{j+2}$.
One can easily check that the distance between pair of points 
in that table is at most one.
 (For an illustration, see Figure \ref{fig:PathA}).

\begin{table}[h]
\centering
\begin{tabular}{ r | l | l  }
  & location of $\CO_A$ & location of $\CO_L$  
 \\
\hline
   if $x_i \in C_j$  & $s_j$ & $M(\Seg{c_js_j})$\\
	& $z_j$ & $c_j$\\ 
	&  & $w_j$\\ 

	& $\Seg{c_jg_j} \dashv \Seg{s_jg_{j+1}} $ & $\Seg{w_{j}\alpha_{j}} \dashv \Seg{c_{j}g_{j}}$ \\

   if $\neg x_i \in C_j$  & $s_j$ & $\Seg{\beta_{j-1}w_j}	\dashv \Seg{c_js_j}$ \\
	& $w_j \perp \Seg{s_jg_{j+1}}$ &$w_j$\\
	& $z_j$ &$z_j$\\
	& &$c_j$\\
	& $\Seg{c_jg_j} \dashv \Seg{s_jg_{j+1}} $ &$ M(\Seg{c_jg_j}) $\\

   if $x_i \notin C_j \& \neg x_i \notin C_j$  & $s_j$ & $\Seg{\beta_{j-1}w_j}	\dashv \Seg{c_js_j}$\\
	&$w_j \perp \Seg{s_jg_{j+1}}$  & $w_j$\\
	&$z_j$  & $z_j$\\

	& &$c_j$\\
	& &$w_j$\\
& $\Seg{c_jg_j} \dashv \Seg{s_jg_{j+1}} $ & $\Seg{w_{j}\alpha_{j}} \dashv \Seg{c_{j}g_{j}}$ \\

\hline
	&  $h_1$ s.t.  $\| h_1\alpha_j \| \le \eps$ & $\alpha_j$\\
	&  	$h_2$ s.t.  $\| h_2 \beta_j \| \le \eps$ & $\beta_j$\\

\hline
if $x_i \in C_{j+1}$ &  $\Seg{s_{j+1}c_{j+1}}	\dashv \Seg{s_{j}g_{j+1}}$    & $\Seg{\beta_jw_{j+1}}	\dashv \Seg{c_{j+1}s_{j+1}}$\\
& $z_{j+1}$ & $w_{j+1}$\\
&  & $z_{j+1}$\\
&  & $c_{j+1}$ \\
&  $g_{j+1}$ & $M(\Seg{c_{j+1}g_{j+1}})$\\

if $\neg x_i \in C_{j+1}$ &  			$\Seg{s_{j+1}c_{j+1}}	\dashv \Seg{s_{j}g_{j+1}}$  & $M(\Seg{s_{j+1}c_{j+1}})$\\
 & $z_{j+1}$ & $c_{j+1}$ \\
 &  & $w_{j+1}$ \\
 &  $g_{j+1}$ & $\Seg{g_{j+1}c_{j+1}}	\dashv \Seg{w_{j+1}\alpha_{j+1}}$ \\

if $x_i \notin C_{j+1} \& \neg x_i \notin C_{j+1}$ &  	$\Seg{s_{j+1}c_{j+1}}	\dashv \Seg{s_{j}g_{j+1}}$  & $\Seg{\beta_jw_{j+1}}	\dashv \Seg{c_{j+1}s_{j+1}}$\\
& $z_{j+1}$ & $w_{j+1}$\\
 & & $c_{j+1}$ \\
 & & $w_{j+1}$ \\
&  $g_{j+1}$ & $\Seg{g_{j+1}c_{j+1}}	\dashv \Seg{w_{j+1}\alpha_{j+1}}$ \\

\hline

	&  $h_3$ s.t.  $\| h_3\alpha_{j+1} \| \le \eps$ & $\alpha_{j+1}$\\
	&  	$h_4$ s.t.  $\| h_4 \beta_{j+1} \| \le \eps$ & $\beta_{j+1}$\\
\hline

 if $\neg x_i \in C_{j+2}$  & $s_{j+2}$ & $\Dir{\alpha_{j+1}w_{j+2}} 	\dashv \Seg{c_{j+2}s_{j+2}}$ \\
   if $ x_i \in C_{j+2}$  & $s_{j+2}$ & $M(\Seg{c_{j+2}s_{j+2}})$\\
   if $x_i \notin C_{j+2} \& \neg x_i \notin C_{j+2}$  & $s_{j+2}$ & $\Dir{\alpha_{j+1}w_{j+2}}	\dashv \Seg{c_{j+2}s_{j+2}}$\\

\end{tabular}
\caption{Distance between pair of points is less or equal to one}
\label{tab:PathA}
\end{table}

Finally, if $k$ is an odd number, then  
$\Dir{s_kv}$ is the last segment along $B$, otherwise, 
$\Dir{g_kv}$ is the last one. In any case, 
that edge crosses  the circle $\CB(\eta,1)$, where $\eta$ is the last vertex of 
$\cfev_i$ before $v$ (point $\eta$ is computed in line \ref{l:ComputeV} of 
Algorithm \ref{alg:reduction}). Therefore, 
 $\CO_A$ can walk to $v$, while keeping distance $1$ to $\CO_L$.

\qed
\end{proof}

\begin{figure}[t]
	\centering
	\includegraphics[width=0.8\columnwidth]{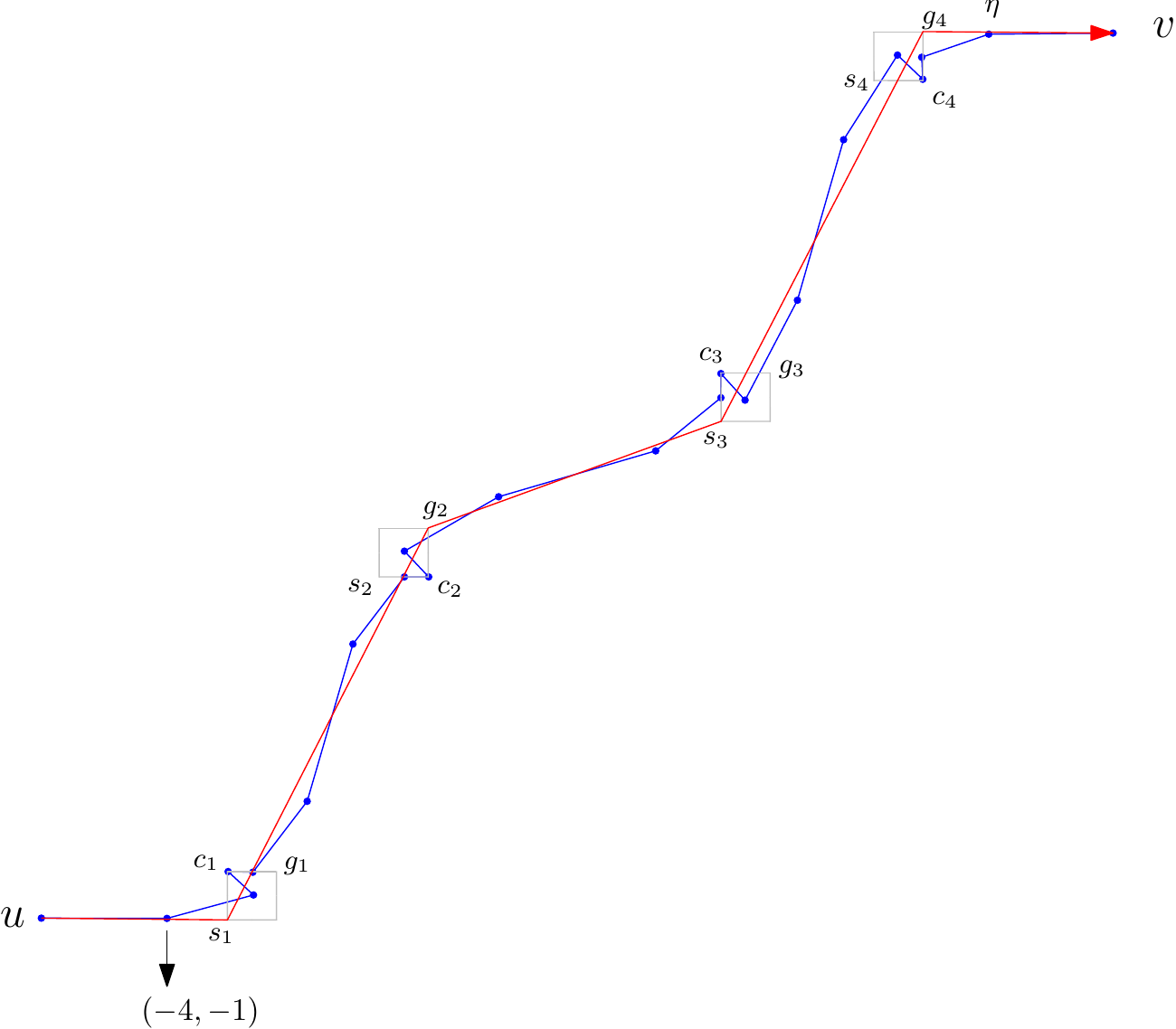}
	\caption{Assume that formula $\phi$ has four clauses $C_1, C_2, C_3$ and $C_4$, 
where the occurrence of variable $x_i$  in those clauses is:
$\neg x_i \in C_1$, $ \neg x_i \in C_2$, $x_i \in C_3$ and $x_i \in C_4$.  
For each clause $C_i$, the reduction algorithm places three point $s_i,g_i$ and $c_i$ in the plane. Blue curve is an example of curve $\cfev_i$ which corresponds to variable $x_i$ in $\phi$. Red curve is curve $A$. 
 }
	\label{fig:pathAExample}
\end{figure}

\begin{lemma}\label{lemma:PathB}
Consider any subcurve $\cfev_i =<u,...,v>$, $1\le i \le n+2$,  
constructed through lines \ref{l:mainstart} to \ref{l:subcurve} 
of Algorithm \ref{alg:reduction}. Let $B$ be the polygonal curve  $<u\gre_1\sma_2\gre_3\sma_4..v>$. Then, $\distF(\cfev_i,B) \le 1$.
\end{lemma}

\begin{proof}
The proof is analogous to the proof of Lemma \ref{lemma:PathA}, see appendix.
\end{proof}

\begin{lemma}\label{lemma:NoSwitchFromAtoB}
Consider any curve $\cfev_i \subset P $, $1\le i \le n+2$, and 
a point object $\CO_L$ walking from $u$ to $v$ on $\cfev_i$. 
Also, imagine two point objects $\CO_A$ and 
$\CO_B$, walking on curves $A$ and $B$ 
(from Lemmas \ref{lemma:PathA} and \ref{lemma:PathB}), respectively
while keeping distance $1$ to $\CO_L$. Then,
if $\CO_A$ switches to path $B$ or $\CO_B$ switches to path $A$, 
they loose distance $1$ to $\CO_L$.
\REM{
take any path to a vertex on $B$ and 
stay in distance 1 to $\CO_L$. Similarly, 
$\CO_B$ can not take any path to a vertex on $A$ and 
stay in distance 1 to $\CO_L$.}
\end{lemma}

\begin{proof}
See Appendix.
\end{proof}

For $1 \le i \le k$, 
if $i$ is an odd number, set $a_i = \sma_i$ 
and $b_i = \gre_i$  and
if $i$ is an even number, 
set $a_i = \gre_i$ 
and $b_i = \sma_i$.

\begin{lemma}\label{lemma:ABCanSeeC}
Consider the curve $A = <ua_1a_2a_3 \dots a_k v>$ from Lemma \ref{lemma:PathA}. Let 
$A_1$ be a subcurve of $A$ which starts at $u$ and ends at $a_j$, $1 \le j \le k$.
Furthermore, let $A_2$ be a subcurve of $A$ which starts at $a_j$ and ends at $v$.
For any curve $\cfev_i$ , $1\le i \le n+2$,
if $x_i \in C_j$, 
$\distF(A_1 \ap c_j \ap A_2, \cfev_i) \le \eps$. 
Similarly, consider the curve $B = <ub_1b_2b_3 \dots b_k v>$ from Lemma \ref{lemma:PathB}. Let 
$B_1$ be a subcurve of $B$ which starts at $u$ and ends at $b_j$, $1 \le j \le k$.
Furthermore, let $B_2$ be a subcurve of $B$ which starts at $b_j$ and ends at $v$.
For any curve $\cfev_i$ , $1\le i \le n+2$,
if $\neg x_i \in C_j$,   
$\distF(B_1 \ap c_j \ap B_2, \cfev_i) \le \eps$. 

\end{lemma}

\begin{proof}
When $x_i$ appears in clause $C_j$, point $z = M(\Seg{c_ja_j})$ is 
a vertex of $\cfev_i$. 
Since $\| c_ja_j \| = 2$ and $z$ is the  midpoint of 
 $\Seg{c_ja_j}$,
$\CO_L$ can wait at $z$ while $\CO_A$ visits $c_j$
When $\neg x_i$ appears in clause $C_j$, point $z = M(\Seg{c_jb_j})$ is 
a vertex of $\cfev_i$. 
Since $\| c_jb_j \| = 2$ and $z$ is the  midpoint of 
 $\Seg{c_jb_j}$,
$\CO_L$ can wait at $z$ while $\CO_A$ visits $c_j$
and comes back to $b_j$
\end{proof}

\begin{lemma}\label{lemma:NOTABCanSeeC}
Consider curve $A$ (respectively, $B$).
For any curve $\cfev_i$, $1\le i \le n+2$,
if $ x_i \notin C_j$ and $\neg x_i \notin C_j$,   
then $A$ (resp., $B$) can not be modified to visit $c_j$.
\end{lemma}
\begin{proof}
This is because $dist(w_j,\Seg{a_jc_j}) >1$
and $dist(w_j,\Seg{b_jc_j}) >1$.
\end{proof}

%For example, 
%in a formula $\phi$ where
%$\neg x_i \in C_1$, $ \neg x_i \in C_2$, $x_i \in C_3$ and $x_i \in C_4$,  
%$\distF( \cfev_i , <u,s_1,g_2,s_3,g_4,v>  )$

\vspace{0.5 in}

\begin{theorem}
Given a formula $\phi$ with $k$ clauses $C_1, C_2, \dots, C_k$ and $n$ variables $x_1, x_2 \dots, x_n$,
as input let curve $P$ and pointset $\pset$ be the output of Algorithm \ref{alg:reduction}. 
Then, $\phi$ is satisfiable iff a 
curve $Q \in Curves(S)$ exists such that 
$\distF(P,Q) \le 1$.
\end{theorem}

\begin{proof}

%please wait .. the references needs to be corrected 

For $(\Rightarrow)$: 
Assume that  formula $\phi$ is satisfied. 
In Algorithm \ref{alg:buildQ}, we show that 
knowing the truth value of the literals in $\phi$, 
we can build a curve $Q$ which 
visits every point in $\pset$ and $\distF(P,Q) \le 1$.

\begin{algorithm} [h]
\caption {{\sc Build a feasible curve $Q$ }} 
\label{alg:buildQ}
\algsetup{indent=1.5em}
\begin{algorithmic}[1]	
		\baselineskip=0.9\baselineskip
	\REQUIRE  Truth table of variables $x_1, x_2, \dots, x_n$ in $\phi$

	\STATE $Q \eq \emptyset$
	\STATE $Q \eq Q \ap t$ \label{l:startPoint}
	 
	\FOR {$i=1$ to $n$}   
	\IF {$(x_i = 1)$}
	\STATE $\pi \eq <ua_1a_2a_3 \dots a_kv>$
	\FORALL {$C_j$ clauses, if $x_i \in C_j$ }
	\STATE let $\pi_1 $ be  subcurve of $\pi$ from $u$	 to $a_j$
	\STATE let $\pi_2 $ be  subcurve of $\pi$ from $a_j$	 to $v$
	\STATE $\pi \eq \pi_1 \ap c_j \ap \pi_2$  \label{l:visitCone}
	\ENDFOR
	\STATE $Q \eq Q \ap \pi$ \label{l:x1}
	\ELSE 	
	\STATE $\pi \eq <ub_1b_2b_3 \dots b_kv>$
	\FORALL {$C_j$ clauses, if $\neg x_i \in C_j$ }
	\STATE let $\pi_1 $ be  subcurve of $\pi$ from $u$	 to $b_j$
	\STATE let $\pi_2 $ be  subcurve of $\pi$ from $b_j$	 to $v$
	\STATE $\pi \eq \pi_1 \ap c_j \ap \pi_2$ \label{l:visitCzero}
	\ENDFOR
	\STATE $Q \eq Q \ap \pi$ \label{l:x0}

	\ENDIF
	\STATE $Q \eq Q \ap t$
%	\STATE $\pi \eq <ua_1a_2a_3 \dots a_kv>$

	\ENDFOR
	\STATE $Q \eq Q \ap <ua_1a_2a_3 \dots a_kv>$\label{l:nplusone}
	\STATE $Q \eq Q \ap t$

	\STATE $Q \eq Q \ap <ub_1b_2b_3 \dots b_kv>$\label{l:nplustwo}
	\STATE $Q \eq Q \ap t$  \label{l:endPoint}

	\STATE {\bf return} {\sc Q}  
\end{algorithmic}
\end{algorithm}

%Imagine a point object $\CO_P$ walking on $P$. 
%and another point object $\CO_Q$ 
%walking on $Q$. 
%Recall that each curve $\cfev_{i} \subset P$ corresponds to 
%a variable $x_i$ in $\phi$.

First we show $\distF(P,Q) \le 1$, where $Q$
is the output curve of Algorithm \ref{alg:buildQ}.
Recall that by Algorithm \ref{alg:reduction}, 
curve $P$ includes $n$ subcurves $\cfev_{i}$ each corresponds 
to a variable $x_i$. 
Both 
curves $P$ and $Q$ start and end at a same point $t$.
For each curve $\pi$ which is appended to $Q$ 
in the $i$-th iteration of Algorithm \ref{alg:buildQ} 
(line \ref{l:x1} or line \ref{l:x0}), 
$\distF(\pi,\cfev_i) \le 1$  by Lemma \ref{lemma:ABCanSeeC}. 
Notice that $P$ also includes two additional subcurves $\cfev_{n+1}$ and $\cfev_{n+2}$ whereas there is no variable $x_{n+1}$ and $x_{n+2}$ in formula $\phi$. These two curves are to resolve two special cases: 
when all variables $x_i$ are 1,  no $\neg x_i$ appears in $\phi$,
and when all variables $x_i$ are 0,  no $x_i$ appears in $\phi$.
Because of these two curves, 
we add two additional curves in line \ref{l:nplusone}
and \ref{l:nplustwo} to $Q$. Finally, by  Observation 
\ref{obs:concat}, $\distF(P,Q) \le 1$.
%when  $\CO_P$  traverse $\cfev_{n+1}$ and $\cfev_{n+2}$, 
%we can visit all $a_j$ and $b_j$ points in pointset $\pset$.  
%Therefore, we only need to show that every  $c_j$ points 
%in $\pset$ can be also visited.

Next, we show that curve $Q$ visits every point in $S$. First of 
all, by the curves added to $Q$ 
in line \ref{l:nplusone} and \ref{l:nplustwo}, 
all $a_j$ and $b_j$, $1\le j \le k$, in $S$ will be visited. 
It is sufficient to show that $Q$ will visit all $c_j$ points in $S$  as well.
Since  formula $\phi$ is satisfied, every clause $C_i$ in $\phi$ must be satisfied 
too. Fix clause $C_j$. At least one of the literals in $C_j$
must have a truth value $1$. If $x_i \in C_j$ and $x_i = 1$, 
then by line \ref{l:visitCone}, curve $Q$ visits $c_j$.
On the other hand, if $\neg x_i \in C_j$ and $x_i = 0$, 
by line \ref{l:visitCzero}, curve $Q$ visits $c_j$. We conclude that 
curve $Q$ is feasible.

Now we prove the $(\Leftarrow)$ part:

Let $Q$ be a feasible curve with respect to $P$ and pointset $\pset$.
Notice that curve $P$ consists of $n$ subcurve $\cfev_i$, 
$1 \le i \le n $, each corresponds to one variable $x_i$. 
From the configuration of each $\cfev_i$ in c-squares, 
one can easily construct formula $\phi$ with 
all of its clauses and literals. 

%Since $Q$ is feasible, 
%all $c_j$ are visited by curve $Q$.

Imagine two point objects $\CO_Q$ 
and $\CO_P$ walk on $P$ and $Q$ respectively. 
We find the truth value of variable $x_i$ in the formula
by looking at the path that $\CO_Q$ takes to stay in $1$-\Frechet distance to $\CO_P$, 
when $\CO_P$ walks on curve $\cfev_i$ corresponding to $x_i$.
If $\CO_Q$ takes path $A$ from Lemma \ref{lemma:PathA} 
while $\CO_P$ is walking on  $\cfev_i$, then $x_i = 1$; 
whereas If $\CO_Q$ takes path $B$ from Lemma \ref{lemma:PathB} 
while $\CO_P$ is walking on  $\cfev_i$, then $x_i = 0$. 
Object $\CO_Q$ decides between path $A$ or $B$, 
, when both $\CO_Q$ and $\CO_P$ are at point $u$. 
Lemma \ref{lemma:NoSwitchFromAtoB} ensures that  
once they start walking, 
$\CO_Q$ can not change its path from $A$ to $B$ 
or from $B$ to $A$. 
Therefore, the truth value of a variable $x_i$ is consistent.

The only thing left to show is the reason that formula $\phi$ is satisfiable. 
It is sufficient to show every clause of $\phi$ is satisfiable. 
Consider any clause $C_j$.
Since curve $Q$ is feasible, 
it uses every point in $\pset$.  
Assume w.l.o.g. that $\CO_Q$ visits $c_j$ 
when $\CO_P$ is walking along curve $\cfev_i$.  
By Lemmas 
\ref{lemma:NoSwitchFromAtoB} and \ref{lemma:ABCanSeeC},
this only happens when either ($x_i$ appears in $C_i$ and $x_i = 1$)
or ($\neg x_i$ appears in $C_i$ and $x_i =0$). 
Therefore, $C_j$ is satisfiable.
%By our reduction algorithm, satisfiability maps to 
%visibility, meaning that if clause $C_j$ is satisfied in 
%formula $\phi$, point $c_j$ in the pointset $S$ can be visited 
%by a feasible path $Q$.

%For $\CO_P$
%to be able to walk on $\cfev_i$, $\CO_Q$ must walk on 
%either path $\CA$ or $\CB$. 
%Now when $\CO_P$  is walking on $\cfev_i$, 
%we check $\CO_Q$ to see whether it
%is walking on path $\CA$ or $\CB$.
%If $\CO_Q$ walks on $\CA$ then, $x_i = T$, otherwise, $x_i = F$.
 %\ldots

%
%Points $a_j$ are all can be visited by a curve 
%$<ua_1a_2a_3..a_kv>$, when subcruve $\cfev_{i+1}$ must be walked. 
%(as there is no variable $x_{i+1}$ in $\phi$). 
%Similarly, points $b_j$ are all can be visited when 
%subcurve $\cfev_{i+2}$ must be walked.
%(as there is no variable $x_{i+2}$ in $\phi$). 
%

%
%Consider variable $x_i$ in $\phi$ and curve 
%$\cfev_i$ in $P$ corresponding to $x_i$. 
%
%if $x_i$ = True, Let $\pi_3 =<ua_1a_2a_3 \dots a_kv>$, 
%
%if $x_i$ = False, Let $\pi_4 =<ua_1a_2a_3 \dots a_kv>$, 
%
%otherwise, let $\pi_4 =<ub_1b_2b_3 \dots b_kv>$. 
%If $x_i$ appears in $ C_j$, Lemma \ref{lemma:ABCanSeeC} ensures that 
%we can modify $\pi_3$ to visit point $c_j$.
%%On the contrary, 
%path $\pi_4$ can not visit a point  $c_j$ 
%whose corresponding clause $C_j$ contains $\neg x_i$.

The last ingredient of the NP-completeness proof is
to show that the reduction takes polynomial time.  
One can easily see that Algorithm \ref{alg:reduction}
has running time $O(nk)$, 
where $n$ is the number of variables in 
the input formula with $k$ clauses.

\end{proof}

%\subsection{Implementation Results}
To show the correctness of above lemmas, we have implemented our reduction 
algorithm.  We test our algorithm on a 
formula $\phi$ with four clauses.  
This enables us to check all possible 
configurations of $\cfev_i$ in Algorithm \ref{alg:reduction}. 
The program generates three sets, a pointset 
$S = \{ s_1,g_1,c_1, s_2,g_2,c_2, s_3,g_3,c_3, s_4,g_4,c_4 \}$, 
a curve set $L$ and a curve set $C$ as follows.

Imagine a polygonal curve which starts from point $u$, 
goes through points in $S$ and ends at $v$. Our 
program generates all possible such curves and keep them in set $C$.
Therefore, $C$ contains almost 1,000,000,000
curves.

\REM{
Each curve $\mu$ in $C$ is built as follows:

\begin{table}[h]
\centering
\begin{tabular}{ l | c }
  permutations & curves
 \\
\hline
 $C_1$ = all size one permutations of points in $S$ & $\mu_1 = <u, C_1,  v>$\\ 
 $C_2$ = all size two permutations of points in $S$  & $\mu_2 = <u,, C_2,  v>$\\
 $C_3$ = all size three permutations of points in $S$  & $\mu_3 = <u,, C_3, v>$\\
		...&...\\
 $C_{12}$= all size twelve permutations of points in $S$ & $\mu_{12} = <u, C_{12},  v>$ \\
\end{tabular}
\vspace{0.2 in}
%\caption{Curves in C}
%\label{tab:noSwitch}
\end{table}

$$ C = \bigcup _{i=1}^{12} \mu_i$$

Therefore, the number of curves in $C$ is:
}

%$$1 + 12 + 132 + 1320 + 11880 + 95040 + \dots  + 479001600$$

Another set $L$ includes all different configuration of curve $\cfev_i$ 
which corresponds to variable $x_i$ in the formula.
Since $x_i$ or $\neg x_i$ or none could appear in one clause 
and the formula has four clauses, the set $L$ contains 81 different curves.

In our application, we compute \Frechet distance between every curve in $C$ and every curve in  $L$. 
The results show that all the curves in $C$
have \Frechet distance greater than $1$ to curves 
in $L$ except two curves $<u,s_1,g_2,s_3,g_4,v>$
and $<u,g_1,s_2,g_3,s_4,v>$. In other words,  
for only 162 pairs of curves, we got:

 $\forall \mu \in L: \distF( \mu , <u,s_1,g_2,s_3,g_4,v>  ) \le 1$  and
 $\forall \mu \in L: \distF( \mu , <u,g_1,s_2,g_3,s_4,v>  ) \le 1$  

In addition to above tests, we verified the correctness of 
Lemma \ref{lemma:ABCanSeeC} in different cases.

\newcommand{\bounC}[1]{{\sigma(\Pol,{#1})}}
\section {Convex Polygon Case}
\label{sec:SpecialCase}

In this section, we address the following %variant of the 
problem:
given a convex polygon $\Pol$ and a pointset $\pset$ in
the plane, 
find a polygonal curve $Q$ whose vertices are from $\pset$,
and the \Frechet distance between $Q$
and a boundary curve of $\Pol$ is minimum.
Note that %in our problem definition, 
each point of $\pset$ must be used in $Q$ and it  
can be used more than once. 
In the decision version of the problem,
we want to decide if there is a polygonal curve $Q$ through all points in  
$\pset$, 
whose \Frechet distance to a polygon's boundary curve 
is at most $\eps$, for a given $\eps \gee 0$. 
%An instance of the decision problem is illustrated in Figure~\ref{fig:instance}.

%This section is organized as follows. First, 
%in the next subsection, we 
%we introduce some notations and describe
%an algorithm from \ref{cccg11} which is 
%the base of our algorithm. Next, 
%in Section \ref{}, 

%To solve the decision problem for a
%convex polygon, we first study the same problem
%for the easier case,
%when  the input is a convex chain. This is because every convex polygon can be decomposed into 
%two convex chains.

\subsection{Preliminaries}
\label{subsec:preliminaries}

We borrow some notations 
from \cite{cccg11}
as we make use of the algorithm in that paper
to solve the decision version of our problem.
%, and $d \gee 2$ be a fixed integer.
For any point $p$ in the plane,
we define $\CB(p,\eps) \equiv \{q \in \IR^2 : \|pq\| \lee \eps\}$
to be a \emph{ball} of radius $\eps$ centered at $p$,
where $\|\cdot\|$ denotes  Euclidean distance.
Given a line segment $L \subset \IR^2$,
we define $\CC(L, \eps) \equiv \cup_{p\in L} \CB(p,\eps)$
to be a \emph{cylinder} of radius~$\eps$ around $L$.

In this section, whenever we say polygon, we mean 
a convex polygon. Also, when we 
say a curve visits a point $u$, we mean that 
$u$ is a vertex of the curve.
We denote by $\CH(\pset)$ the convex hull of pointset $\pset$.

%
%
%A curve  in $\IR^d$ can be represented as  a continuous function 
%$P:[0,1] \rightarrow \IR^d$.
%Given two points $u,v \in P$,  
%we write $u \lei v$, if $u$ is located before $v$ on $P$.
%
For an interval $I$ of points in  $P$,
we denote by $\Left(I)$ and $\Right(I)$
the first and the last point of $I$ along $P$, respectively.
Given two points $u$ and $v$ in $\pset_i$, 
we say $u$ is {\em before\/} $v$, % (or {\em to the left of\/} $q$), 
and denote it by $u \lex v$ when
$\Left(P_i[u])$ is located before $\Left(P_i[v])$ on $P$.
Moreover, 
we say $u$ is entirely before $v$ (or $v$ is entirely after $u$) 
and denote it by $u \lei v$, when $\Right(P_i[u])$ is located before $\Left(P_i[v])$ on $P$.

\vspace{0.3 in}

\hspace{-0.2 in}{\bf The Decision Algorithm in \cite{cccg11} }
Given a polygonal curve $P$ of size $n$ (with 
starting point $s$ and ending point $t$), 
a pointset $\pset$ of size $k$
and a distance $\eps>0$, 
the algorithm in that paper decides in $O(nk^2)$ time,
whether there exists a polygonal curve through some points 
in $\pset$ in $\eps$-\Frechet distance to 
$P$. 
%Let $u$ be a point in $\pset$ and $u'$
%be a point on the boundary of the polygon such that
%$\| uu' \| \le \eps$.
%Let $\PST$ denote the starting and ending point 
%of a boundary curve $P$.
%Imagine polygonal curve $P$ with  $\PST = \Ps$.

Curve $P$ is composed of $n$ 
line segments $\PolSeg_1 \ldots \PolSeg_{n}$. 
For each segment $\PolSeg_i$,
$C_i$ denotes the cylinder $\CC(\PolSeg_i, \eps)$,
and $\pset_i$ denotes the set $\pset \cap C_i$.
Furthermore, for each point $v \in C_i$,
$\PolSeg_i[v]$  denotes the line segment $\PolSeg_i \cap \CB(v,\eps)$ \cite{cccg11}.
A polygonal curve $R$ is called \emph{semi-feasible} if all its vertices are from $\pset$ and 
$\distF(R, P') \lee \eps$ for a subcurve $P' \subseteq  P$  starting at $s$.
A point $v \in \pset_i$ is called \emph{reachable}, at cylinder $C_i$,
if there is a semi-feasible curve ending at $v$ in $C_i$.

Following is a brief outline of how the decision 
algorithm in \cite{cccg11} works: 
it processes the cylinders one by one from 
$C_1$ to $C_{n}$, and identifies at each cylinder $C_k$
all points of $\pset$ which are reachable at $C_k$.
The reachable points for each cylinder $C_k$, $k$ from 1 to $n$, is maintained in a set
called {\em reachability set}, denoted by $\R_k$.
In the $k$-th iteration of the algorithm,
first all points in $\pset_k$
which are reachable through a point in a set 
$\R_i$, for $1 \lee i < k$, are added to $\R_k$ .
These points are called the \emph{entry points} of cylinder $C_k$. 
We denote, by $\lme{k}$, the leftmost 
entry point of cylinder $C_k$, which 
at this step, equal to $q = \min_{v \in \R_k}{\Left(P_k[v])}$.
Next,  all points in $\pset_k$
which are reachable through $\lme{k}$
are added to $\R_k$.
Finally, the decision algorithm 
returns YES , if  
$\R_n \cap \CB(t,\eps) \not= \emptyset$.

%------------------------- Decision Algorithm -------------------------------

\begin{figure}[t]
	\centering
	\includegraphics[width=0.6\columnwidth]{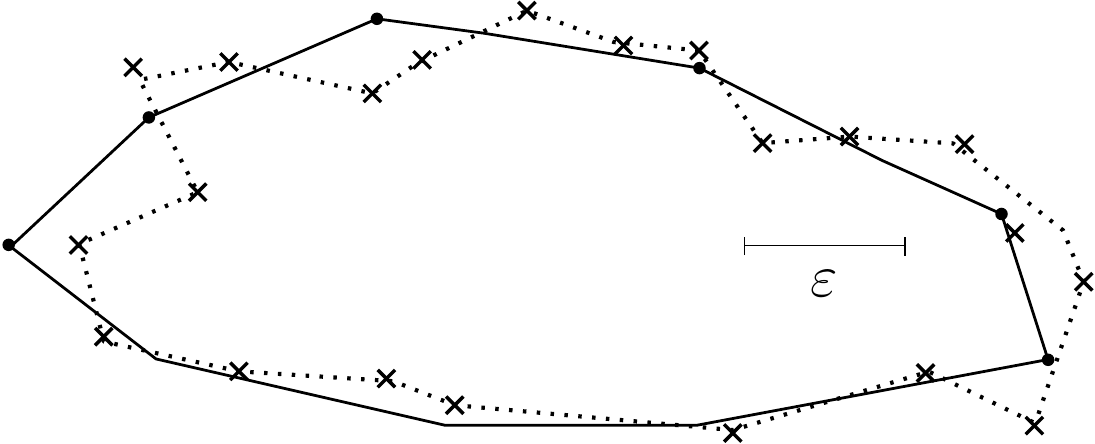}
	\caption{A problem instance}
	\label{fig:types}
\end{figure}

\subsection {Decision Algorithm}%------------------------------
\label{subsec:decAlg}

Let $\Pol$ be a convex polygon of size $n$, 
$S$ be a pointset of size $k$ and $\eps \ge 0$ be an input distance. 
For $\Pol$, we call a curve, {\em a boundary curve} of $\Pol$, denote it by $\bounC{u}$,
if it starts from point $u$ on the boundary of $\Pol$, 
goes around the polygon on the boundary once and ends at 
$u$.

\begin{definition} \label{def:feasible}
Given a pointset $S$, a convex polygon $\Pol$ and a distance $\eps$,  
a polygonal curve $Q$,
is called {\em feasible} if 
$Q \in Curves(S)$ and a boundary curve $\bounC{z}$ exists such that 
$\distF(\bounC{z},Q) \le \eps$.
\end{definition}

As the first step of our algorithm, 
we execute the decision algorithm in \cite{cccg11}. 
Note that here as opposed to in \cite{cccg11}, 
the starting and ending point of 
the curve is unclear because the input is a convex polygon.
Which point on the boundary we choose?
The following lemma justifies our choice later:

\begin{lemma} \label{lema:convexhull}
Given a convex polygon $\Pol$, a pointset $\pset$ and a distance $\eps$,
%assume that a feasible curve exists through $\pset$. 
a necessary condition to have a feasible 
curve through $S$ is: 
$\distF(\sigma(\CH(\pset),z), \bounC{z'} ) \le \eps$, 
where $z$ is a vertex of $\CH(\pset)$ and $z'$ is a point on the boundary of $\Pol$ s.t. $\|zz'\| \le \eps$.
\end{lemma}

\begin{proof}
See Appendix.
\end{proof}

Let $\Qs$ be the point with 
smallest $x$-coordinate in $\pset$
and $\Ps$ be a point on the boundary of $\Pol$
s.t.   $\|\Qs\Ps\| \le \eps$. Furthermore, 
let $\rho = \bounC{\Ps}$ (curve $\rho$ has $n+1$ line segments
$\rho_1, \rho_2, \dots ,\rho_{n+1}$) 
Run the decision algorithm in \cite{cccg11}, 
with parameters $\rho, \pset$ and $\eps$.
Result is the reachability sets 
$\R_1, \R_2, \dots ,\R_{n+1}$ where each $\R_i$ maintains the reachable points at cylinder $C_i$.

%Assume that after the execution of the algorithm in \cite{cccg11},
%every point $v \in \pset$ is at least in some $\R_i$.

%Consider two consecutive cylinders $C_i$ and $C_j$, $j>i$,
%where $j$ is the smallest index such that $\pset_j \neq \phi$.
%We say, $C_j$ is the next cylinder to $C_i$.

%From now on, assume that every point of $\pset$
%is in some $\R_i$.

\begin{definition}
\label{def:Types}
%Let $u$ be a point in 
%We say, $C_j$ is the next cylinder to $C_i$.

Consider  two consecutive reachability 
sets  $\R_i$ and $\R_j$, $ 1 \le i<j\le n$.
Let $u$ be a point in $\R_i$. Then, we call $u$
a $\Good$ point at $C_i$ if  there exists a semi-feasible
curve which contains $u$ as its vertex
and ends at $\lme{j}$.
Otherwise, we call $u$ a $\SemiBad$ point at $C_i$ 
and we call $\lme{j}$ the rival of $u$. 
\end{definition}

\begin{obs} \label{obs:orderABC}
%\begin{lemma} \label{obs:orderABC}
Let $u$ and $v$ be $\Good$ and $\SemiBad$  points at cylinder $C_i$, 
respectively. Then, $u \lei v$.
\end{obs}

\begin{obs}
%\begin{lemma}
A $\SemiBad$ point $u$ at cylinder $C_i$ is:
\begin{itemize}
\item a $B1$ point at  $C_i$: when $u$ 
doesn't reach to $\lme{j}$ but it reaches to 
some other point $v$ in $\R_j$ 
(see Figure \ref{fig:types}a). 

\item a $B2$ point at $C_i$,
when $u$ does not reach to 
any point in $\R_j$, 
but it reaches to a point in $\R_k$,
$i<j \le k$ (see Figure \ref{fig:types}b).

\item a $B3$ point at $C_i$,
when $u$ does not reach to 
any point in $\R_k$, $i<j \le k$ (see Figure \ref{fig:types}c).

\end{itemize}

\end{obs}

\begin{figure}[t]
	\centering
	\includegraphics[width=0.9\columnwidth]{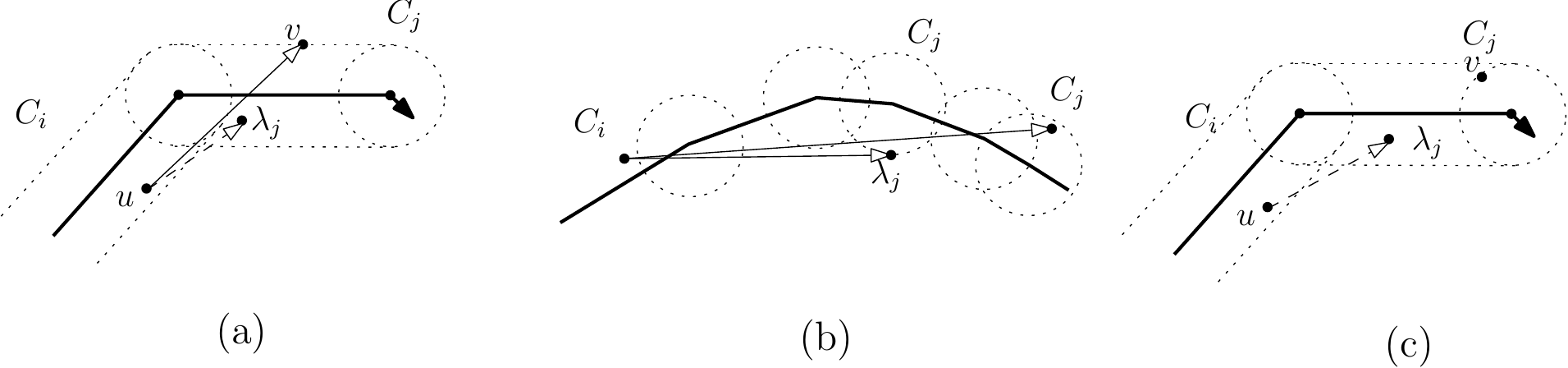}
	\caption{(a) Point $u$ is a $B1$ point at $C_i$ 
(b) Point $u$ is a $B2$ point at $C_i$ (c) Point $u$ is a $B3$ point at $C_i$.}
	\label{fig:types}
\end{figure}

After running the algorithm in \cite{cccg11},
we process the reachability sets $\R_1,\R_2, \dots \R_{n-1}$, one by one in order, 
and we identify the types of the points for every point in each set. 
Notice that a point $u$ may be located in 
multiple cylinders, so it might be reachable 
at more that one cylinder and thus
be in more than one reachability set.
%For the sake of simplicity, 
%we keep the types of point $u \in \pset$ (from Definition \ref{def:Types}) 
%in a set denoted by $T(u)$. 
%For example, if $u$ is
%of $\Good$ at $C_i$, $T(u)$ contains $\TA_i$ 
%and if $u$ is $\SemiBad$ point at $C_i$, 
%$T(u)$ contains $\TB_i$. 

% 
%does not reach to any point in next cylinders.  
%That means a $\Bad$ point can be visited 
%by a feasible curve only at its ending point. Therefore, 
%we can only have one $\Bad$ point when
%a feasible curve exists.

%\begin{proof}
%Assume, for the sake of contradiction, that  
%for a point $v\in \pset$,  $v$ is not in any $\R_i$ 
%but a feasible curve $Q$ exists through points of $\pset$.

%Running Time
%
%\begin{theorem}
%Given a convex polygonal chain $P$ of size $n$, a pointset $S$ of size $k$ and a distance $\eps$, we can decide in 
%$O(nk^2)$ time whether there exists a feasible curve through $S$ or not.
%\end{theorem}
%

\newcommand{\TwiceB}{Twice-TypeB}
\newcommand{\Upper}{\pi}
\newcommand{\Lower}{\mu}

%Consider a polygonal curve $\rho = \Upper \oplus \Lower$.
%%$P_1\dotsP_n\dotsP_1 \dotsP_m\dotsP_1$
%We run Algorithm  \ref{alg:Cons} with parameters $\rho$, $\pset$
%and $\eps$.
%This gives us two curves,  one is 
%$\rho_a$ in which all $\SemiBad$ points are skipped and 
%another one is $\rho_b$ in which all $\SemiBad$ are visited.  
%in $\eps$-\Frechet distance to $\pi$
%and another two curves $\Lone$ and $\Ltwo$ 
%in $\eps$-\Frechet distance to $\mu$.

\newcommand{\Rival}{Rival}

%Let $\Pol$ be a convex polygon. 
Let $\Upper$ and $\Lower$ be the upper 
and the lower chain of $\Pol$, respectively.  
Let Tube($\Pol$) be the union of all $\CC(\PolSeg_i, \eps)$, 
where $P_i$ is an edge of $\Pol$, $1 \le i \le n$.
Tube($\Upper$) and Tube($\Lower$) are defined, analogously.

%Let $\Pol$ be a convex polygon. 

\begin{definition}
We call a point a  \TwiceB ~point, 
if it is of $\SemiBad$ at two cylinders.
%Type $B$ at a cylinder corresponding to an edge in $\Upper$
%as well as at a cylinder corresponding to an edge in $\Lower$.
\end{definition}

\begin{definition}

We call a connected area within Tube($\Pol$),
shared between two cylinders, one corresponding 
to an edge in $\Upper$ and another corresponding to 
an edge in $\Lower$, a \DoubleB~area, 
if it contains a \TwiceB ~point. 
\end{definition}

\begin{lemma}\label{lemma:twice}
Given a polygon $\Pol$, a pointset $\pset$ 
and $\epsilon>0$, at most two \DoubleB~areas  exists.  
\end{lemma}

\begin{proof}
%Let $\Upper$
%and $\Lower$ be the upper and the lower chain of $\Pol$,
%respectively.  

Assume w.l.o.g. that $\Pol$ is stretched horizontally (see Figure \ref{fig:twice})(In the case that $\Pol$ is vertically stretched, decompose 
it into right and left chains and the rest of 
argument is the same as here).

Assume for the sake of contradiction,
that there are three such areas and 
a point $w$ is a \TwiceB ~point located in the middle one.
%Assume that a point $w$ is a \TwiceB ~point.
(see Figure \ref{fig:twice}). 
Let $w'$ and $w''$ be the rivals of $w$ in the upper 
and lower chain of $\Pol$, respectively. 
Assume w.l.o.g. that $w$ and $w'$
are located in two consecutive cylinders (within Tube($\Upper$)) which 
share $\CB(p_1,\eps)$. 
Similarly, assume w.l.o.g. that $w$ and $w''$
are located in two consecutive cylinders (within Tube($\Lower$)) which 
share $\CB(p_2,\eps)$.
Since $w$ is a $\SemiBad$ point and $w'$ is the rival of $w$, 
edge $\Dir{ww'}$ does not cross circle $\CB(p_1,\eps)$. 
Because of the same reason, edge $\Dir{ww''}$ does not cross circle $\CB(p_2,\eps)$. 
This implies that, vertices of $\Pol$ to the right of $p_1$ has 
a $y$-coordinate less than $p_1$ and vertices before $p_2$  
has a $y$-coordinate greater than $p_2$, meaning that $p_2$
has the lowest $y$ coordinate among vertices of $\Pol$.
Therefore, no \DoubleB~area exists to the right of the 
area in which $w$ is located, a contradiction.

%
%The proof of the lemma is now easily seen. 
%One can easily seen that one of the rivals of $u$, 
%is for sure located inside the band around the 
%upper and the lower chains.  This point can not 
%be then $\SemiBad$ in the upper chain because the 
%area 

\REM{
To prove the lemma, we first investigate a situation where 
a point is in more than one reachability sets and then from there, we establish
the lemma.

If $u \in \pset$ is in $\R_{i}$ and $\R_{i+1}$, 
then it is $\Good$ point at $C_i$.
Let $p$ be a vertex of the polygon and $\CB(p,\eps)$ be 
the circle shared between $C_i$ and $C_{i+1}$.

One of these two cases may happen here: (i) point $u$ lies in
$\CB(p,\eps)$ (see Figure \ref{fig:twice}a), 
then obviously $u$ is a $\Good$ point at $C_i$
(ii)  point $u$ doesn't lie in
$\CB(p,\eps)$ (see Figure \ref{fig:twice}b), we show 
that although $\Dir{u\lambda_{i+1}}$ could not 
cross $\CB(p,\eps)$, 
but still $u$ is a $\Good$ point at $C_i$:

Assume for the sake of contradiction, that 
$u$ is $\SemiBad$ point at $C_i$.
Thus, by the definition of $\SemiBad$ points, 
no semi-feasible 
curve exists which contains $u$ as its vertex
and ends at $\lme{i+1}$.
Therefore, since $u \in \R_{i+1}$, 
it must be located entirely after 
$\lme{i+1}$ in $C_{j+1}$
(for example, in Figure \ref{fig:twice}b, $u$ 
can be located at the place where $u'$ is). This contradicts 
the fact that $u$ is in $S_i$(or located within $C_i$).
}
\end{proof}

\begin{figure}[t]
	\centering
	\includegraphics[width=1\columnwidth]{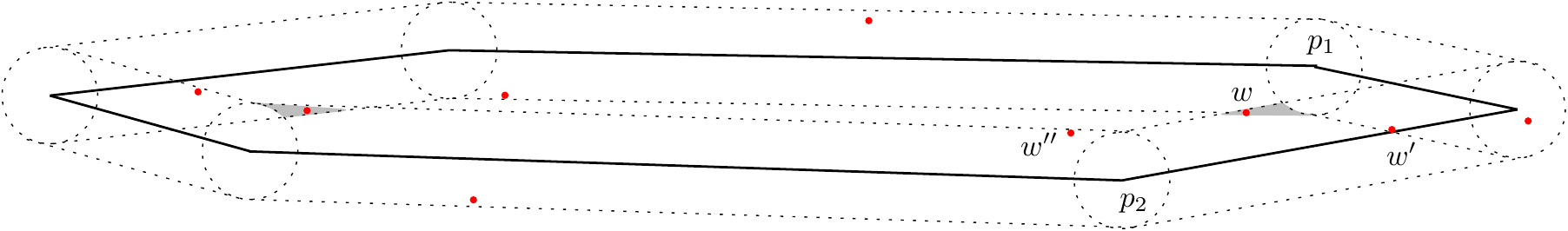}
	\caption{Proof of Lemma \ref{lemma:twice}}
	\label{fig:twice}
\end{figure}

%\begin{col}
%If $x$ is a $\SemiBad$ point in $\Lower$, then it is $\Good$ point in $\Upper$,
%If $x$ is a $\SemiBad$ point in $\Upper$, then it is $\Good$ point in $\Lower$.
%\end{col}

\begin{lemma} \label{lemma:last}
There exists a feasible curve iff Algorithm \ref{alg:dec} returns
YES.
\end{lemma}
\begin{proof}
First  $(\Leftarrow)$:
It is easy to observe that 
conditions in lines  
\ref{l:everyReach},
\ref{l:last},
\ref{l:B3} are necessary
to have a feasible curve 
through $\pset$.
It suffices to show for the curve 
$\alpha$ built by Algorithm \ref{alg:dec}, $\distF(\alpha, \rho) \le \eps $.

We show this by induction on the 
number of edges in $\rho$.
To handle the base case of the induction, assume that 
$\rho$ has an additional edge consist 
of only point $x'$.
Imagine a point object $\CO_\Pol$
walking on the boundary of $\Pol$, starting from 
point $x'$ and imagine another point object $\CO_\alpha$
walking on curve $\alpha$ starting from $x$ while 
keeping distance $\le 1$ to  $\CO_\Pol$.
Since distance $\CO_\Pol$ and $\CO_\alpha$ is less 
than $\eps$ at the start,
the base case of the induction holds. 

We show $\CO_\Pol$ can walk the whole $\rho$ 
and ends at $x'$.
Assume inductively that in the loop in line \ref{l:forCyl},
we have processed $\R_1$ to  $\R_{n-1}$, 
and now we are about to process $\R_{n}$.
So every point in $\pset_{i-1}$ is in 
$\alpha$ and $\CO_\Pol$ can walk to a point 
in $\rho_{i-1}$ in $\eps$ distance to $\CO_\alpha$.
Let $v$ be the leftmost point at $C_i$  such that 
$end(\alpha) \see v$.
If $v = \lme{i}$, by Observation \ref{obs:concat}, 
we can add  every point in $\pset_i$ to $\alpha$ 
and $\CO_\Pol$ can proceed to $\rho_i$. Otherwise;
	since $\lme{i}$ is a reachable point, a point 
	$u$ must exist in $\R_k$, $k<i$ such that 
	$u$ reach $\lme{i}$. Assume that 
	$u$ is after all points in direction $\Dir{\rho}$
	which can reach to $\lme{i}$.
	It suffices to show point $u$ is a vertex of $\alpha$, 
	so that in line \ref{l:rightmost}, by reaching to $u$, 
	the algorithm stops removing vertices in $\alpha$
	and connects $u$ to $\lme{i}$.
	Two cases happen here: 
     (i) $u \in \pset_{i-1}$, and  some points in $\pset_{i-1}$ are 
	$\Good$ and some are $\SemiBad$ points.
    Then, $\SemiBad$ points are removed 
	from $\alpha$ (because they can not reach $\lme{i}$ by definition), 
	and $u$ will be connected to	$\lambda_i$ (see Figure \ref{fig:skipped}).
	Therefore, $\CO_\Pol$ can walk to the next edge on $\rho$. 
     (ii) When $u \in \R_k$, $k < i-1$, then observe that because of 
the convexity of the polygon, $u$ reach
points in $\R_j$ which can not reach $\lme{i}$, $k<j<i$ (see Figure \ref{fig:skipped}). 
Therefore, $u$ is a vertex in $\alpha$. 
%Notice that in the last case we discussed, 
%$u$ is a $\Good$ point at $C_k$. 
%However, if $u$ be type $B$ point, 
%for sure it must be of type $B3$. 
%It is only left to show that, what happens 
%if we have a point of type $B3$ among the points. 

Now we show that curve $\alpha$ can be modified such 
that it visits every \TwiceB ~point when $\CO_\Pol$ walks on $\Upper$
(with the same argument, curve $\alpha$ can be modified to 
visit all such points when $\CO_\Pol$ walks on $\Lower$).
%To do so, first we that we can reach the first \TwiceB ~point, then visit all of those point and finally proceed to the next cylinders on $\Upper$.
Let's call \DoubleB~areas, area I and area II.
Let $b$ be the first $\SemiBad$ point in 
direction $\Dir{\Upper}$ at area I. 
Assume w.l.o.g. that $b$ is in $\pset_{i}$.
Let $u$ be a point after all points in direction $\Dir{\Upper}$
which can reach $\lme{i+1}$.
Because of the convexity of polygon,
point $u$ reaches $b$. 
It is clear that as soon as $\alpha$  reaches $b$, 
it can visit all $\SemiBad$ points 
at $C_i$ in sorted $\lei$ order. 
Let $d$ be the leftmost point 
in $\pset_{i+1}$ which the last 
$\SemiBad$  point at $C_i$ can reach. 
If $d$ is a vertex of $\alpha$, then
$\alpha$ now has all  \TwiceB ~points at area I.
When $d$ is not a vertex of $\alpha$,
because of the convexity of the polygon, 
still $d$ can be connected to a 
vertex of $\alpha$ which is reachable from $d$. 
Therefore, one can modify curve $\alpha$
to visit \TwiceB ~points.

%a point in the next cylinder 
%but not the leftmost entry point, 
%that point has for sure more view from the 
%leftmost point reachable in that cylinder.

$(\Rightarrow)$:
Let $Q$ be a feasible curve through $S$.
Assume that there is  no \TwiceB ~point
among points in $\pset$. Let $z$ be any point in $\pset$.
Then, by Lemma \ref{lemma:twice}, if 
$z$ lies in both $Tube(\Upper)$ and
$Tube(\Lower)$, then it can be $\SemiBad$ only 
with respect to one of $\Upper$ or $\Lower$.
Assume w.l.o.g, that $z$ is a $\SemiBad$ point in the 
upper chain. Therefore, $z$ is a $\Good$ point in the 
lower chain and  it will be visited by 
$\alpha$. Now, if  there are some \TwiceB ~points in $S$,
it is easy to show that  $z$ will be visited  by one of 
$\beta$ or $\gamma$ curves in our algorithm.

\REM{
\ref{} $\CO_\alpha$  and $\CO_P$
walk along their path, keeping distance $\eps$
to each other, until 
$\CO_C$  reaches to vertex $v$ in $\CH(S)$
and $\CO_P$ reaches to point $v' \in e_i$ on the boundary of $\Pol$
where $\|vv'\| \le \eps$.
}

\REM{
we can construct a feasible curve:
Lemma \ref{lemma:twice} ensures if a point 
which is $\SemiBad$ in the upper chain, 
can not be  $\SemiBad$  again in the lower chain
and vice versa, a point 
which is $\SemiBad$ in the lower chain, 
can not be  $\SemiBad$  again in the upper chain;
unless it is located within $\DoubleB$ area.

Therefore, we start from point $z$, process 
cylinders one by one and visit all $\Good$ 
points in each cylinder. If we encounter a $\SemiBad$ point, 
we ignore it, because that point will become $\Good$ point
later and we can visit it. 
Let $area I$ and $area II$ denote 
\DoubleB areas.
Now, to take care of \TwiceB ~points, 
build these two curves $Q_1, Q_2$:
\begin{itemize}
\item $Q_1$: when walking on $\pi$, ignore rival of $\SemiBad$ in $area I$
and $area II$ and visit those $\SemiBad$ points. 
\item $Q_2$: when walking on $\mu$, ignore rival of $\SemiBad$ in $area I$
and $area II$ and visit those $\SemiBad$ points.
\end{itemize}

Return YES, if any of $Q_1$ or $Q_2$ visits every point in $\pset$;
otherwise, return NO. 
}
\end{proof}

\begin{figure}[t]
\centering
	\includegraphics[width=0.8 \textwidth]{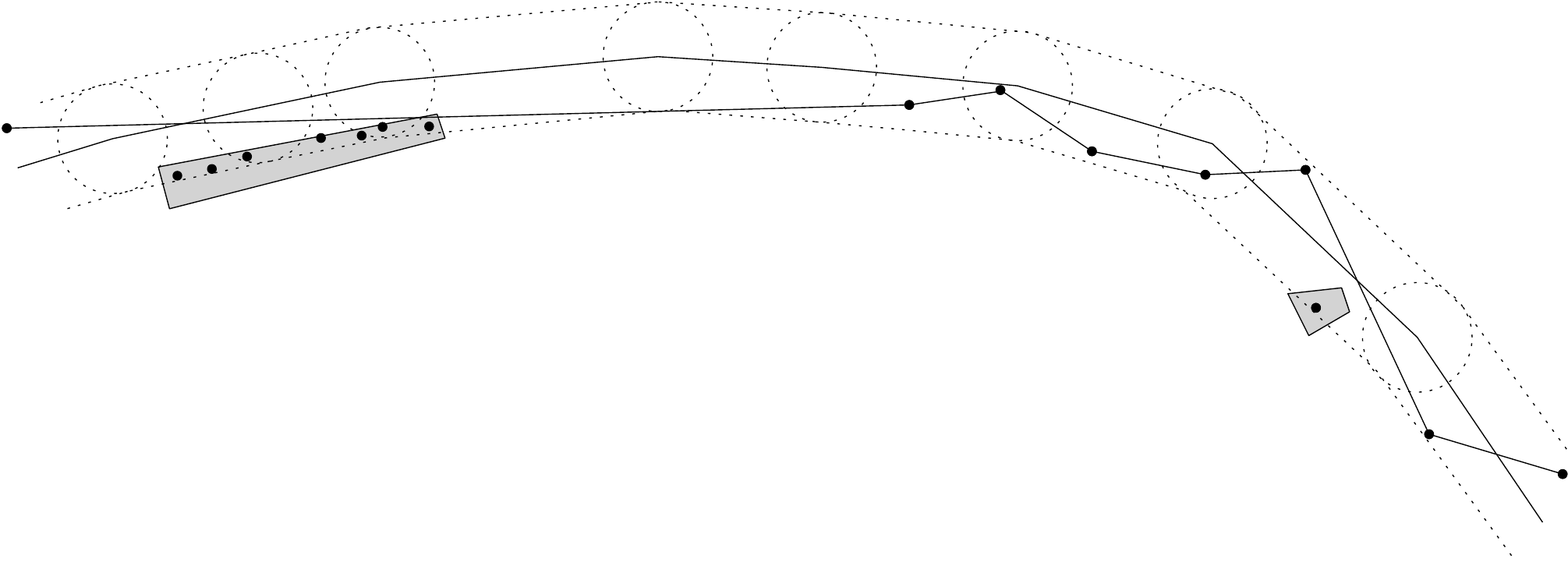}	
	\caption{Proof of Lemma \ref{lemma:last}, points in the 
shaded area can not reach the leftmost entry point of
the next non-empty cylinder }	
	\label{fig:skipped}
\end{figure}

\begin{algorithm} [h]
\caption {{\sc Decision Algorithm }} 
\label{alg:dec}
\algsetup{indent=1.5em}
\begin{algorithmic}[1]	
		\baselineskip=0.9\baselineskip

	\vspace{0.3em}
	\REQUIRE  A convex polygonal $\Pol$, pointset $\pset$ and distance $\eps$

\STATE let $x$ be the point with 
smallest $x$-coordinate in $\pset$
\STATE let $x'$ be a point on the boundary of $\Pol$
s.t.   $\|xx'\| \le \eps$

\STATE let $\rho = \bounC{x'}$ (curve $\rho$ has $n+1$ line segments) 

\vspace{0.1 in}
	\STATE run the decision algorithm in \cite{cccg11}, 
	with parameters $\rho, \pset$ and $\eps$
	\STATE {\bf return} {\sc no} if a point in $\pset$ 
is not in any $\R_{i}$, $1\le i \le n+1$  \label{l:everyReach}

	\STATE {\bf return} {\sc no} if $\R_{n+1} \cap \CB(x',\eps) = \emptyset$ \label{l:last}
	\STATE process reachability sets from $\R_1$ to $\R_n$ in order, \\to identify types of points in each $\R_i$ \label{l:identifytypes}
	\STATE {\bf return} {\sc no} if a point exists of type  $B3$ \\ in every cylinder at which that point is reachable  \label{l:B3}

%\STATE {\bf return} {\sc no} if a $\SemiBad$ point $u$ found 
%which is not in $Tube(\Upper) \cap Tube(\Lower)$ 
\label{l:typeB}

\vspace{0.2 in}
%	\STATE  $\alpha, \ctwo \eq \emptyset$ \label{l:init1}
	\STATE  $\alpha \eq \emptyset$ \label{l:init1}
	\STATE  $\alpha \eq \alpha \ap x$ \label{l:init1}
 	\FOR {$i = 1$ to $n$ } \label{l:forCyl}	
			\STATE let $v$ be a leftmost point at $C_i$  s.t.  $end(\alpha) \see v$  \label{l:leftmost}	
	\IF { $v \neq \lme{i}$}	 \label{l:con}	
			\STATE remove $end(\alpha)$ until a vertex $u$ of $\alpha$ is found s.t. $u \see 		\lme{i}$ \label{l:rightmost} \\
	\ENDIF
	\STATE 
$\alpha \eq \alpha \ap$ 
(every point in $\pset_i$ in sorted $\lei$ order)
%add  every point in $\pset_i$ 
%				to $\alpha$  in sorted $\lei$ order \label{l:addall}

%    \STATE add  every point in $\pset_i$ after $v$
%				to $\ctwo$  in sorted $\lei$ order \label{l:addall}

	\ENDFOR

	\IF {a \TwiceB ~point exists}
	\STATE $\beta \eq$ modify $\alpha$ s.t. it visits those points in $\Upper$ \label{l:lasttwo1}
	\STATE $\gamma \eq$ modify $\alpha$ s.t. it visits those points in $\Lower$
\label{l:lasttwo2}
	\ENDIF
	\STATE {\bf return} {\sc YES} if $\alpha$ or $\beta$ or $\gamma$ include all points in $\pset$

\end{algorithmic}
\end{algorithm}

\begin{theorem} \label{thm:main}
	Given a convex polygon $\Pol$ with $n$ edges and a set $S$ of $k$ points, 
	we can decide in $O(nk^2)$ time whether there is feasible curve $Q$ through $S$ 
	for a given $\eps \gee 0$. Furthermore, a feasible curve $Q$ which minimizing the 
\Frechet distance to the boundary curve of $\Pol$, can be computed in $O(nk^2 \log (nk))$ time.

\end{theorem}

\section{Conclusions}
\label{sec:conc}
In this paper, we investigate the problem of 
 deciding whether a polygonal curve through a given point set
$\pset$ exists which is in $\eps$-\Frechet distance 
to a given curve $P$. We showed that this problem
is NP-Complete. Also, we presented a polynomial 
time algorithm for the special case of the 
problem, when a given curve is a convex polygon. 

Several open problems arise from our work.
From our first result, 
one could investigate some heuristic methods or 
approximation algorithms.
From the second part, in particular, it is interesting 
to study the  problem in the case 
where the input is a monotone polygon, a simple polygon
or a special type of curve.

\nocite{Simpson}

\clearpage
\appendix

\begin{figure}[t]
	\centering
	\includegraphics[width=0.9\columnwidth]{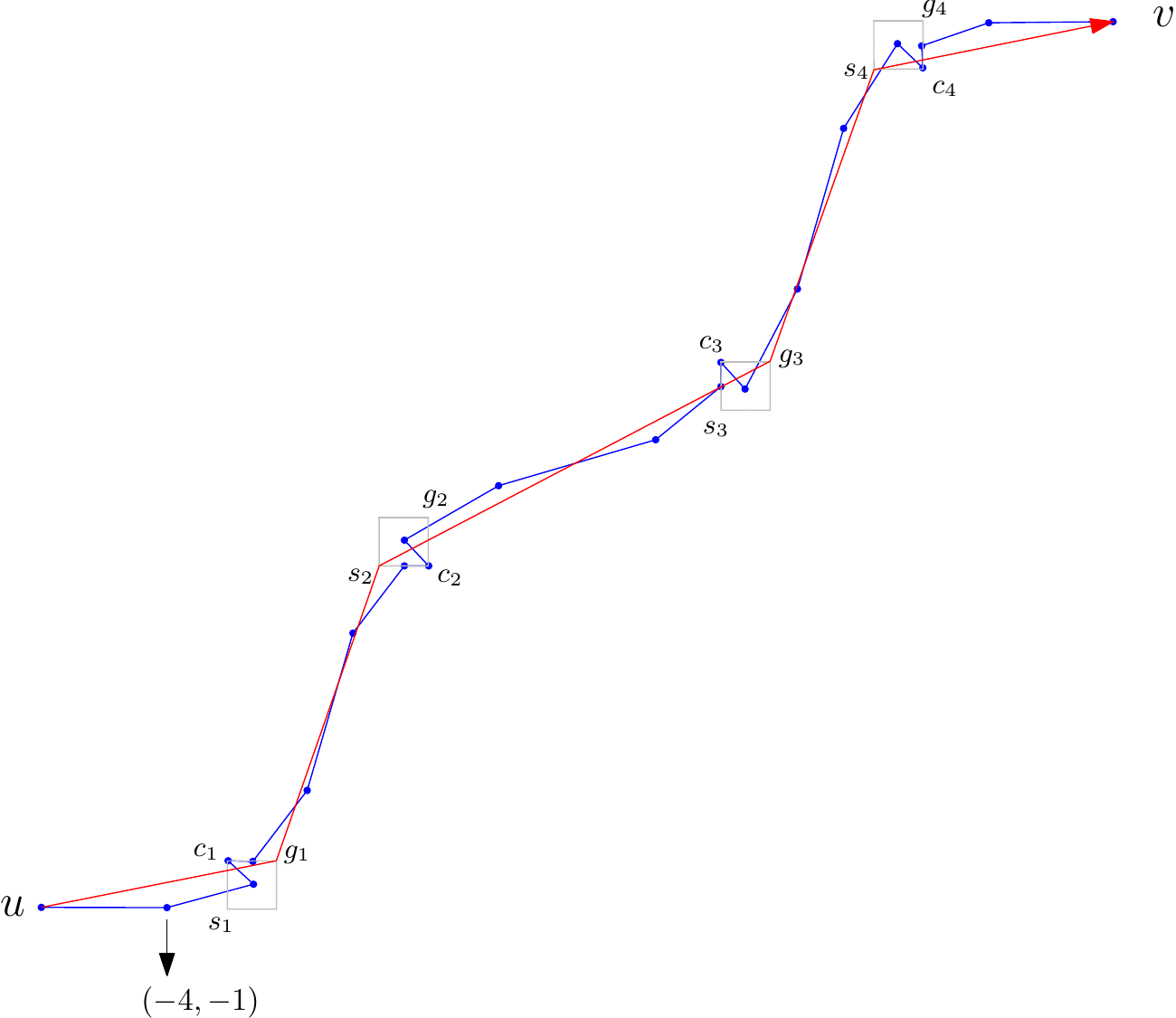}
	\caption{Red curve is curve $B$. Blue curve is an example of curve $\cfev_i$ which corresponds to variable $x_i$ in 
	 formula $\phi$ with four clauses $C_1, C_2, C_3$ and $C_4$. The occurrence of variable $x_i$  in the clauses is: 
$\neg x_i \in C_1$, $ \neg x_i \in C_2$, $x_i \in C_3$ and $x_i \in C_4$.
}
	\label{fig:pathBExample}
\end{figure}

\begin{figure}[h]
	\centering
	\includegraphics[width=0.9\columnwidth]{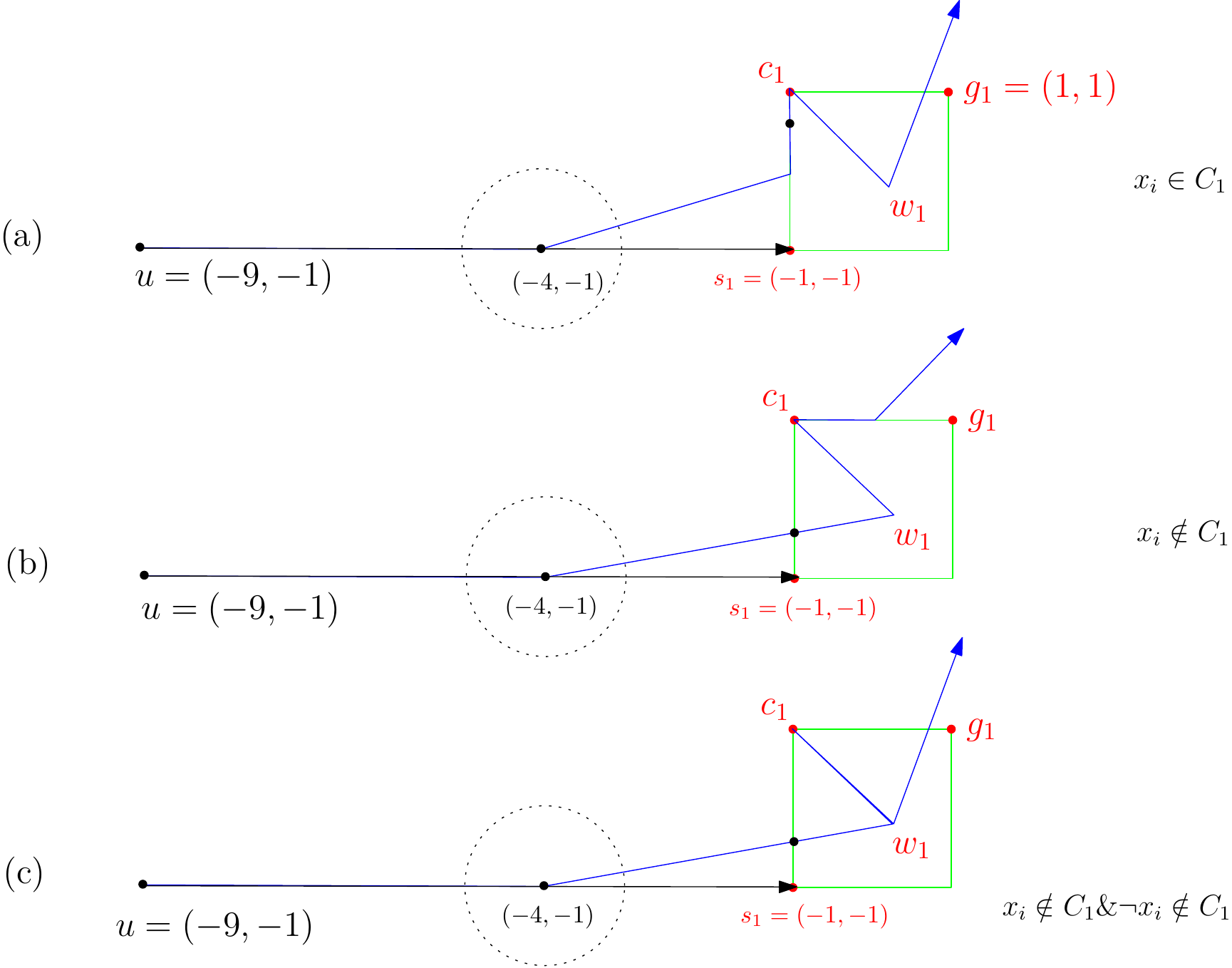}
	\caption{Base case of induction in the proof of Lemma \ref{lemma:PathA}}
	\label{fig:PathAClause1}
\end{figure}

\begin{figure}
	\centering
	\includegraphics[width=0.9\columnwidth]{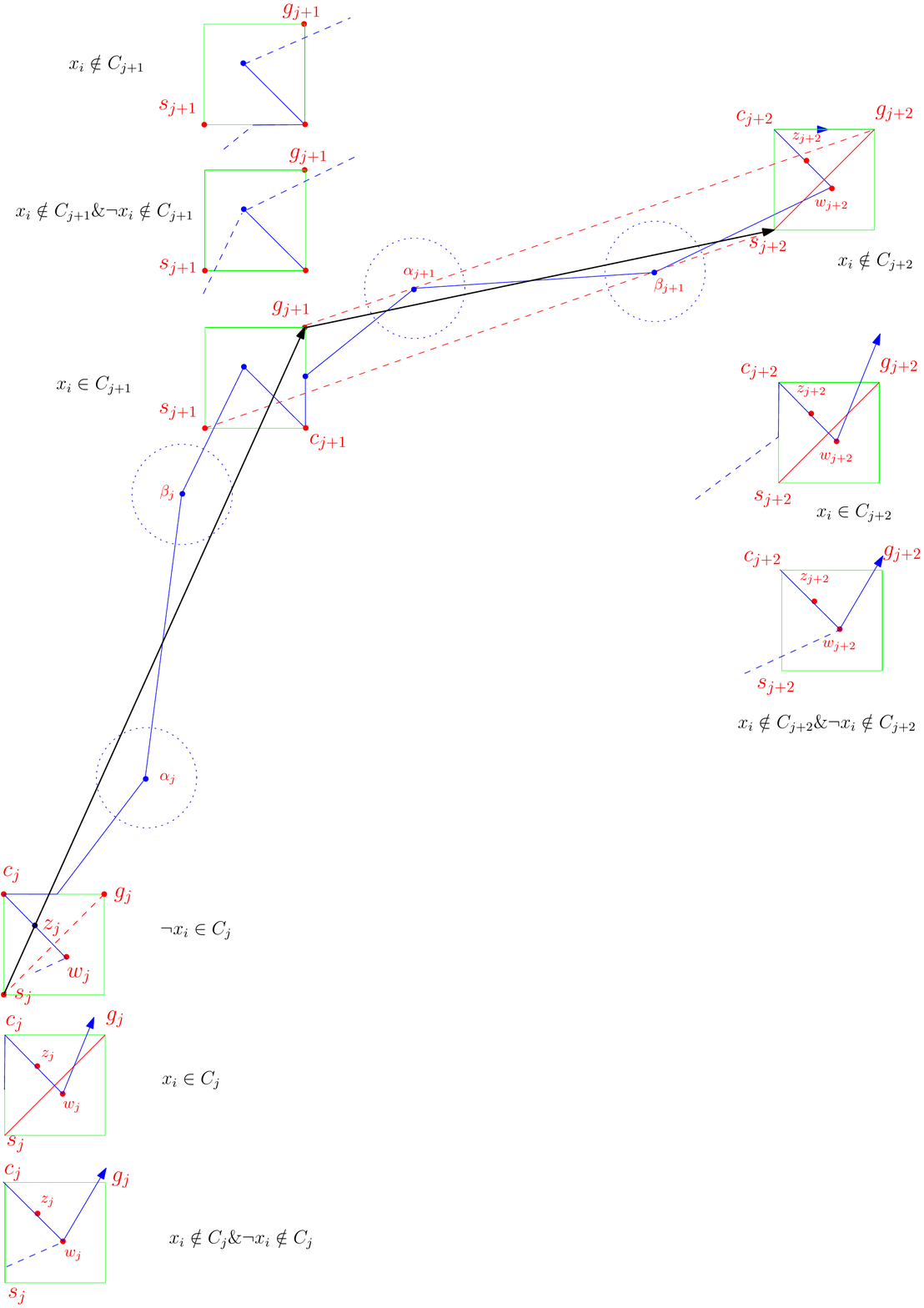}
	\caption{Proof of Lemma \ref{lemma:PathA}}
	\label{fig:PathA}
\end{figure}

\begin{figure}[h]
	\centering
	\includegraphics[width=0.9\columnwidth]{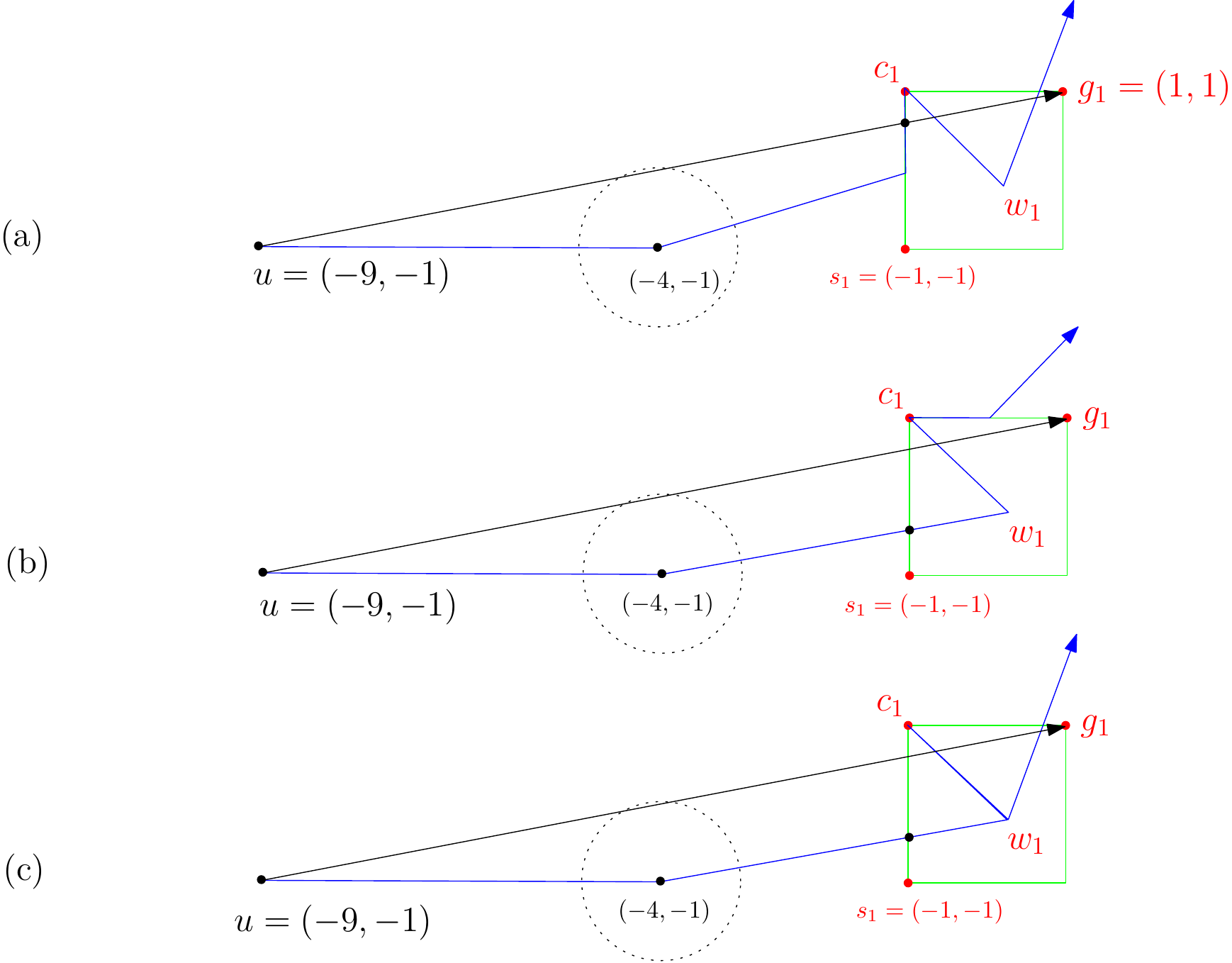}
	\caption{Base case of induction in the proof of Lemma \ref{lemma:PathB} }
	\label{fig:PathBBaseCase}
\end{figure}

\begin{figure}

	\centering
	\includegraphics[width=0.9\columnwidth]{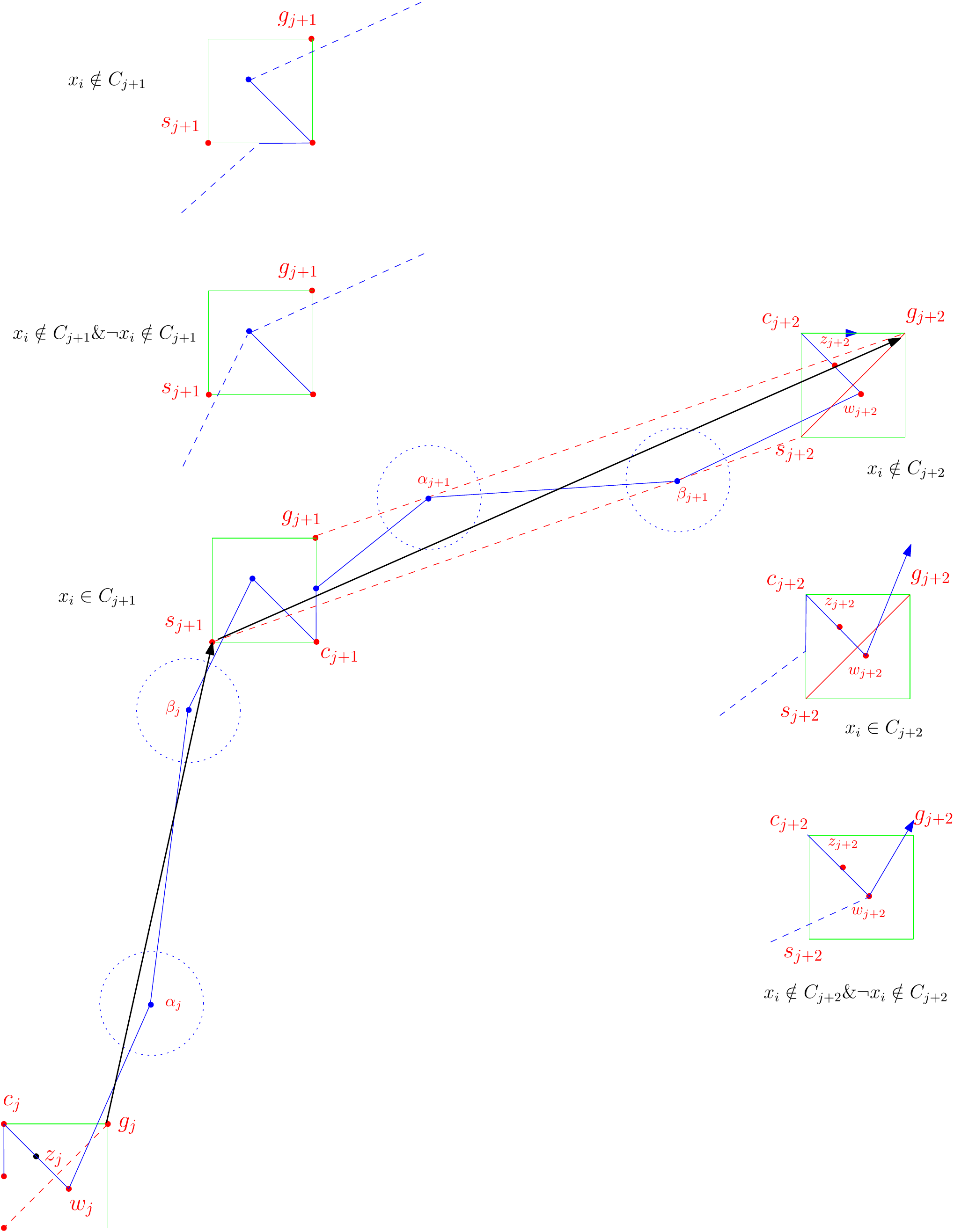}
	\caption{Proof of Lemma \ref{lemma:PathB}}
	\label{fig:PathB}
\end{figure}

\noindent {\bf  Proof of Lemma \ref{lemma:PathB}}:
\begin{proof}
Consider two point objects $\CO_L$ and $\CO_B$ 
traversing $\cfev_i$ and $B$, respectively (Figure \ref{fig:pathBExample} depicts an instance of $\cfev_i$ and $B$).
To prove the lemma, we show that $\CO_L$ and $\CO_B$ can walk
their respective curves, from beginning to
the end, while keeping distance $1$ to each other
. 

The base case of induction holds as follows 
(see Figure \ref{fig:PathBBaseCase} for an illustration):
Table \ref{tab:BaseCasePathB} lists  pairwise location of 
$\CO_L$ and $\CO_B$, where the distance of each pair is less or equal to $1$.
Therefore, $\CO_B$ can walk from $u$ to $g_1$ while
keep distance one to $\CO_L$.

\begin{table}[h]
\centering
\begin{tabular}{ r | l | l  }
if $x_i \in C_1$   & location of $\CO_B$ & location of $\CO_L$  
 \\
\hline
    
&  $u$ & $u$  \\
&  				 $h_1$ s.t.  $\| h_1\mu_1 \| \le \eps$ & $\mu_1 = (-4,-1)$\\

				&    $h_2 = \Dir{ug_1} \dashv	 \Seg{s_1c_1} $ & $\mu_2 = M(s_1c_1)$   \\
&  				 			      $h_2 $ & $c_1$  \\
&    $\Seg{ug_1} \dashv \Seg{c_1w_1}$ &  $\Seg{ug_1} \dashv \Seg{c_1w_1}$ \\
&  				 		     $w_1 \perp ug_1 $ &   $w_1$  \\
& $g_1 $  & $\Seg{w_1\alpha_1}	\dashv \Seg{c_1g_1}$	      \\

\hline
if $\neg x_i \in C_1$ 
&  $u$ & $u$  \\
&  				 $h_1$ s.t.  $\| h_1\mu \| \le \eps$ & $\mu_1 = (-4,-1)$\\
& $h_2 = \Dir{ug_1} \dashv	 \Seg{s_1c_1} $ &  $\mu_2 = \Dir{\mu_1w_1} \dashv	 \Seg{s_1c_1} $\\
& $\Seg{ug_1} \dashv \Seg{c_1w_1}$ & $w_1$ \\
&  &  $c_1$\\
& $g_1$ &  $M(c_1g_1)$\\

\hline
if $x_i \notin C_{1} \& \neg x_i \notin C_{1}$ & $u$ & $u$  \\
&  				 $h_1$ s.t.  $\| h_1\mu \| \le \eps$ & $\mu_1 = (-4,-1)$\\
& $h_2 = \Dir{ug_1} \dashv	 \Seg{s_1c_1} $ &  $\mu_2 = \Dir{\mu_1w_1} \dashv	 \Seg{s_1c_1} $\\
& $\Seg{ug_1} \dashv \Seg{c_1w_1}$ & $w_1$ \\
&  &  $c_1$\\
&  &  $w_1$\\
& $g_1 $  & $\Seg{w_1\alpha_1}	\dashv \Seg{c_1g_1}$	      \\

\end{tabular}
\vspace{0.2 in}
\caption{Pairwise location of $\CO_B$ and $\CO_L$, to prove the base case of induction in Lemma \ref{lemma:PathB} }
\label{tab:BaseCasePathB}
\end{table}

Assume inductively that $\CO_L$ and $\CO_B$ have feasibly walked along 
their respective curves, until $\CO_B$ reached $g_j$.
Then, as the induction step, 
we 
show that
$\CO_B$ can walk to $s_{j+1}$ and then to $g_{j+2}$ 
%along $\Dir{s_jg_{j+1}}$ 
, while keeping distance $1$ to $\CO_L$.
This is shown in Table \ref{tab:PathB}
 (see Figure \ref{fig:PathB} for an illustration).

\begin{table}[h]
\centering
\begin{tabular}{ r | l | l  }
  & location of $\CO_B$ & location of $\CO_L$  
 \\
\hline
   if $x_i \in C_j$  & $g_j$ & $\Dir{\alpha_{j-1}w_j} 	\dashv \Seg{c_jg_j}$ \\
   if $\neg x_i \in C_j$  & $g_j$ & $M(\Seg{c_jg_j})$\\
   if $x_i \notin C_j \& \neg x_i \notin C_j$  & $g_j$ & $\Dir{\alpha_{j-1}w_j}	\dashv \Seg{c_jg_j}$\\

\hline
	&  $h_3$ s.t.  $\| h_3\alpha_j \| \le \eps$ & $\alpha_j$\\
	&  	$h_4$ s.t.  $\| h_4 \beta_j \| \le \eps$ & $\beta_j$\\

\hline
if $x_i \in C_{j+1}$ &  				 $s_{j+1}$  & $\Seg{\beta_jw_{j+1}}	\dashv \Seg{c_{j+1}s_{j+1}}$\\
& $w_{j+1} \perp \Seg{s_{j+1}g_{j+2}}$ & $w_{j+1}$\\
& $z_{j+1}$ & $z_{j+1}$\\

&  $\Seg{g_{j+1}c_{j+1}}	\dashv \Seg{s_{j+1}g_{j+2}}$ & $c_{j+1}$ \\
&  & $M(\Seg{c_{j+1}g_{j+1}})$
 \\

if $\neg x_i \in C_{j+1}$ &  				 $s_{j+1}$  & $M(\Seg{s_{j+1}c_{j+1}})$\\

 & $z_{j+1}$ & $c_{j+1}$ \\
 &  & $w_{j+1}$ \\

 &  $\Seg{g_{j+1}c_{j+1}}	\dashv \Seg{s_{j+1}g_{j+2}}$ & $\Seg{g_{j+1}c_{j+1}}	\dashv \Seg{w_{j+1}\alpha_{j+1}}$ \\

if $x_i \notin C_{j+1} \& \neg x_i \notin C_{j+1}$ &  				 $s_{j+1}$  & $\Seg{\beta_jw_{j+1}}	\dashv \Seg{c_{j+1}s_{j+1}}$\\
& $z_{j+1}$ & $w_{j+1}$\\
 & & $c_{j+1}$ \\
 & & $w_{j+1}$ \\
&  $\Seg{g_{j+1}c_{j+1}}	\dashv \Seg{s_{j+1}g_{j+2}}$ & $\Seg{g_{j+1}c_{j+1}}	\dashv \Seg{w_{j+1}\alpha_{j+1}}$ \\

\hline

	&  $h_5$ s.t.  $\| h_5\alpha_{j+1} \| \le \eps$ & $\alpha_{j+1}$\\
	&  	$h_6$ s.t.  $\| h_6 \beta_{j+1} \| \le \eps$ & $\beta_{j+1}$\\
\hline

if $ x_i \in C_{j+2}$ &  		$\Seg{s_{j+2}c_{j+2}}	\dashv \Seg{s_{j+1}g_{j+2}}$ 		   & $M(\Seg{s_{j+2}c_{j+2}})$\\

 & $z_{j+2}$ & $c_{j+2}$ \\
 &  & $w_{j+2}$ \\

 &  $g_{j+2}$ & $\Seg{g_{j+2}c_{j+2}}	\dashv \Seg{w_{j+2}\alpha_{j+2}}$ \\

if $\neg x_i \in C_{j+2}$ &  				$\Seg{s_{j+2}c_{j+2}}	\dashv \Seg{s_{j+1}g_{j+2}}$ 		   &   $\Seg{s_{j+2}c_{j+2}}	\dashv \Seg{\beta_{j+1}w_{j+2}}$   \\
 & $z_{j+2}$ & $w_{j+2}$ \\
 &  & $c_{j+2}$ \\

 &   $g_{j+2}$ & $M(\Seg{c_{j+2}g_{j+2}})$ \\

if $x_i \notin C_{j+2} \& \neg x_i \notin C_{j+2}$ &  				
$\Seg{s_{j+2}c_{j+2}}	\dashv \Seg{s_{j+1}g_{j+2}}$ 		   &   $\Seg{s_{j+2}c_{j+2}}	\dashv \Seg{\beta_{j+1}w_{j+2}}$  \\
& $z_{j+2}$ & $w_{j+2}$\\
 & & $c_{j+2}$ \\
 & & $w_{j+2}$ \\
&  $g_{j+2}$ & $\Seg{g_{j+2}c_{j+2}}	\dashv \Seg{w_{j+2}\alpha_{j+2}}$ \\

\end{tabular}
\caption{Distance between pair of points is less or equal to one}
\label{tab:PathB}
\end{table}

Finally, if $k$ is an odd number, then  
$\Dir{g_kv}$ is the last segment along $B$, otherwise, 
$\Dir{s_kv}$ is the last one. In any case, 
that edge crosses  circle $\CB(\eta,1)$, where $\eta$ is the last vertex of 
$\cfev_i$ before $v$ (point $\eta$ is computed after the condition checking 
in line \ref{l:ComputeV} of 
Algorithm \ref{alg:reduction}). Therefore, 
 $\CO_B$ can walk to $v$, while keeping distance $1$ to $\CO_L$.

\qed
\end{proof}

\noindent {\bf  Proof of Lemma \ref{lemma:NoSwitchFromAtoB}}:
\begin{proof}
Notice that we have placed $cl_{i+1}$ points far enough from 
$cl_{i}$ points so that 
no curve can go to $cl_{i+1}$
and come back to $cl_i$ and stay 
in $1$-\Frechet distance to $\cfev_i$.
Therefore, to prove the lemma, 
we only focus on two consecutive c-squares.
We show that no subcurve $l' \subseteq \cfev_i$ exists such 
that (for an illustration, see Figure \ref{fig:noswitch}) :

\begin{itemize}

\item $\distF(l',\Dir{\sma_j\gre_j}) \le 1$ because:

\

for all $j$, $1 \le j \le k$, point $c_j$ is always a vertex of $\cfev_i$. 
A point on $\cfev_i$ in distance 1 
to $\sma_j$ lies before $c_j$ 
in direction $\Dir{\cfev_i}$, 
while a point on $\cfev_i$ in distance 1 
to point $\gre_j$ lies after $c_j$ in direction $\Dir{\cfev_i}$.
Since $dist(c_j, \Seg{\sma_j\gre_j}) >1$, 
no subcurve $l' \subseteq \cfev_i$ exists such that 
$\distF(l',\Dir{\sma_j\gre_j}) \le 1$.

\

\item $\distF(l',<\sma_jc_j\gre_{j}>) \le 1$ or $\distF(l',<\gre_jc_j\sma_{j}>) \le 1$, because:

\

For all $j$, $1 \le j \le k$, $w_j$
is a vertex of $\cfev_i$. 
A point on $\cfev_i$ in distance 1 
to $\sma_j$ lies before $w_j$ 
in direction $\Dir{\cfev_i}$, 
while a point on $\cfev_i$ in distance 1 
to point $\gre_j$ lies after $w_j$ in direction $\Dir{\cfev_i}$.
Since $dist(w_j, \Seg{\sma_jc_j}) >1$ and 
$dist(w_j, \Seg{\gre_j\gre_j}) >1$,
no subcurve $l' \subseteq \cfev_i$ exists such that 
$\distF(l', <\sma_jc_j\gre_{j}> ) \le 1$.
Similarly,  no subcurve $l' \subseteq \cfev_i$ exists such that 
$\distF(l', <\gre_jc_j\sma_j> ) \le 1$.

\

\item $\distF(l',<\sma_j\sma_{j+1}>) \le 1$ or  $\distF(l',<\gre_j\gre_{j+1}>) \le 1$ because:

\

Vertex $\alpha_{i}$ of $\cfev_{i}$
guarantees the first part as $dist( \alpha_{i},\Seg{\sma_j\sma_{j+1}}) > 1 $, 
and vertex $\beta_{i}$ of $\cfev_{i}$
guarantees the second part, 
as $dist( \beta_{i},\Seg{\gre_j\gre_{j+1}}) > 1$.

\

\item $\distF(l',<c_jc_{j+1}>) \le 1$,  because $dist( \alpha_{i},\Seg{c_jc_{j+1}}) > 1$

\

\item $\distF(l',<uc_1>) \le 1$, because $dist( (-4,-1),\Seg{uc_1}) >1$ 

\

\item  $\distF(l',<c_j\gre_{j+1}>) \le 1$, because $dist( \alpha_i,\Seg{c_j\gre_{j+1}}) >1$ 

\

\item  $\distF(l',<c_j\sma_{j+1}>) \le 1$, because $dist( \alpha_i,\Seg{c_j\sma_{j+1}}) >1$ 

\

\item  $\distF(l',<c_kv>) \le 1$, because $dist( \eta,\Seg{c_kv}) >1$ 
\end{itemize}

\end{proof}

\begin{figure}
	\centering
	\includegraphics[width=0.9\columnwidth]{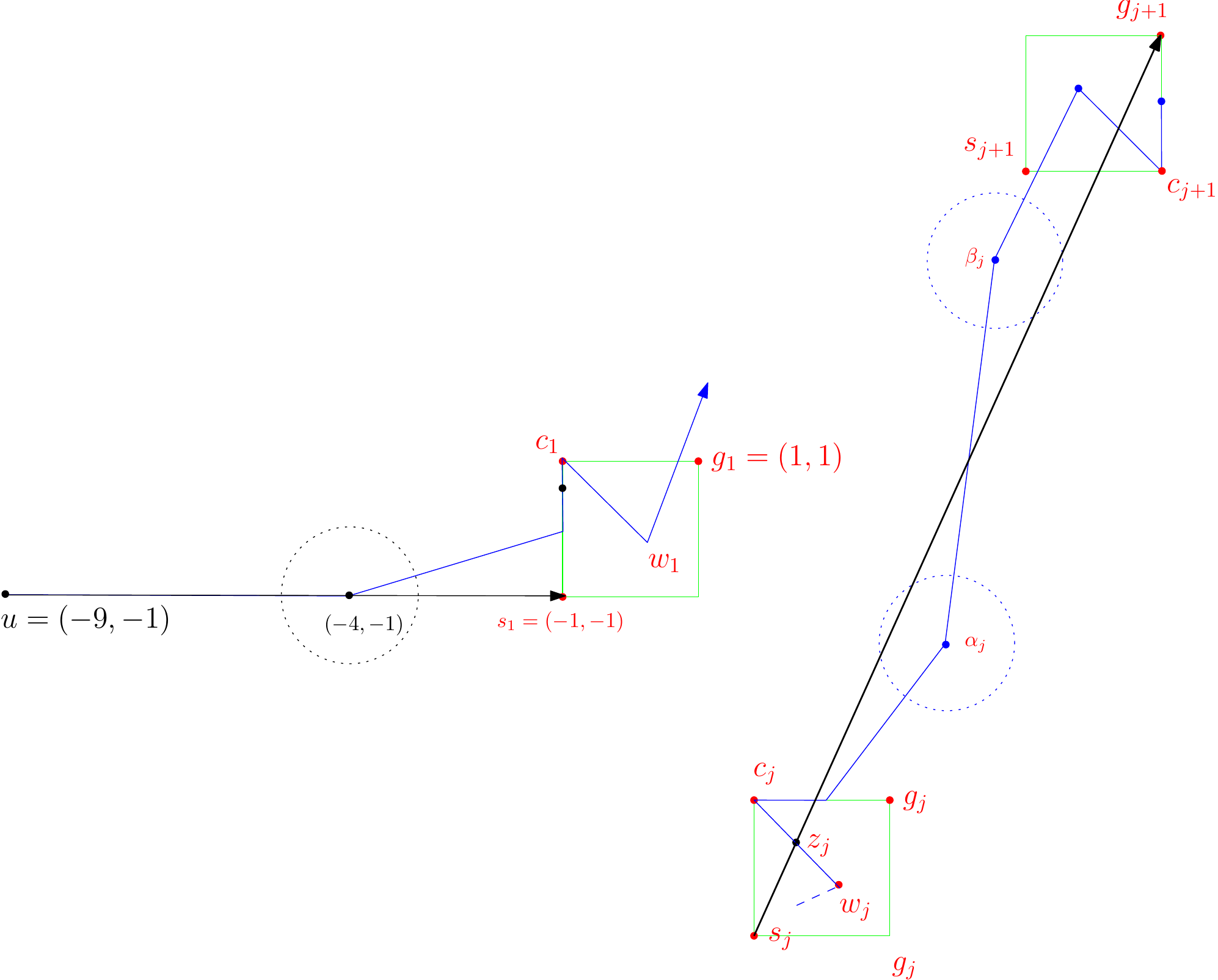}
	\caption{Proof of Lemma \ref{lemma:NoSwitchFromAtoB}}
	\label{fig:noswitch}
\end{figure}

\begin{table}[h]
\centering
\begin{tabular}{ r | l | l  }
if $x_i \in C_1$   & location of $\CO_A$ & location of $\CO_L$  
 \\
\hline
    
&  $u$ & $u$  \\
&  				 $h_1$ s.t.  $\| h_1\mu_1 \| \le \eps$ & $\mu_1 = (-4,-1)$\\

& $s_1 $  & $M(\Seg{s_1c_1})$	      \\

\hline
if $\neg x_i \in C_1$ 
&  $u$ & $u$  \\
&  				 $h_1$ s.t.  $\| h_1\mu_1 \| \le \eps$ & $\mu_1 = (-4,-1)$\\

& $s_1$ &  $ \Dir{\mu_1w_1} \dashv	 \Seg{s_1c_1} $\\

\hline
if $x_i \notin C_{1} \& \neg x_i \notin C_{1}$ & $u$ & $u$  \\
&  				 $h_1$ s.t.  $\| h_1\mu_1 \| \le \eps$ & $\mu_1 = (-4,-1)$\\
& $s_1$ &  $ \Dir{\mu_1w_1} \dashv	 \Seg{s_1c_1} $\\

\end{tabular}
\vspace{0.2 in}
\caption{Proof of Lemma \ref{lemma:PathA}, the base case of induction}
\label{tab:BaseCasePathA}
\end{table}

\noindent {\bf  Proof of Lemma \ref{lema:convexhull}}:

\begin{proof}
%Let $R$ be a boundary curve  of the convex hull of $\pset$. 
We prove the lemma by induction on the number of 
edges in $\CH(\pset)$. 
Let  $e_1e_2..e_n$ be the edges of $\Pol$, numbered 
after an arbitrary vertex of $\Pol$ in clockwise order. 
Obviously, to have a feasible curve, 
every point of $\pset$ must be located within some cylinder $\CC(e_i, \eps)$.
To establish the lemma, we show that when a feasible curve
exists through $\pset$, the condition holds. 

Imagine a point object
%which 
%walks on the boundary of the polygon, 
%and a point object $\CO_C$ which walks on the boundary of the convex hull.
$\CO_C$ which cycles $\CH(\pset)$, 
starting from a vertex of 
$\CH(S)$, say point $u$, and ending at the same point.
Let $u'$ be a point on the boundary of $\Pol$ such that 
$\| uu' \| \le \eps$.
Imagine another point object $\CO_P$  
which starts from $u'$, 
walks on the boundary of $\Pol$ until it reaches the same point $u'$.
Both of the objects walk clockwise. 

To handle the base case of the induction, assume that 
the convex hull has an additional edge consist 
of only  point $u$. Since the distance between
$\CO_C$ and $\CO_P$ is less 
than $\eps$ at the start,
the base case of the induction holds. 
Assume inductively that $\CO_C$  and $\CO_P$
walk along their path, keeping distance $\eps$
to each other, until 
$\CO_C$  reaches to vertex $v$ in $\CH(S)$
and $\CO_P$ reaches to point $v' \in e_i$ on the boundary of $\Pol$
where $\|vv'\| \le \eps$.

Let $w$ be the next  vertex  on $\CH(S)$
after $v$. If $w$ is located within the same 
cylinder $\CC(e_i, \eps)$ as $v$, 
then by Observation \ref{obs:concat}, the lemma holds. 
Otherwise, assume that 
$w$ is in $\CC(e_j, \eps)$, $i<j$ and 
$w'$ is a point in $e_j$ where $\|ww'\| \le \eps$. 
Consider the $\eps$-ball around each of the vertices of 
the polygon between points $v'$ and $w'$. 
It suffices to show that edge $e = \Dir{vw}$ crosses 
each of those balls.

Assume, for the sake of contradiction, that edge $e$ 
does not cross one of those balls, say the one around 
vertex $p$. 
Assume w.l.o.g. that all points 
in $\pset$ lie to the right of $e$.
Then, two cases occur as illustrated in Figure \ref{fig:convexHull}: 
case (a), where $\CB(p,e)$ lies to the right of $e $, which contradict 
the fact that the polygon is convex;
or case (b), where $\CB(p,e)$ lies to the left of $e $, which contradicts 
our assumption that a feasible curve exists, 
because
no point in $\pset$ is in $\eps$-distance to vertex $p$.
\qed
\end{proof}

\begin{figure}{t}
	\centering
	\includegraphics[width=0.9\columnwidth]{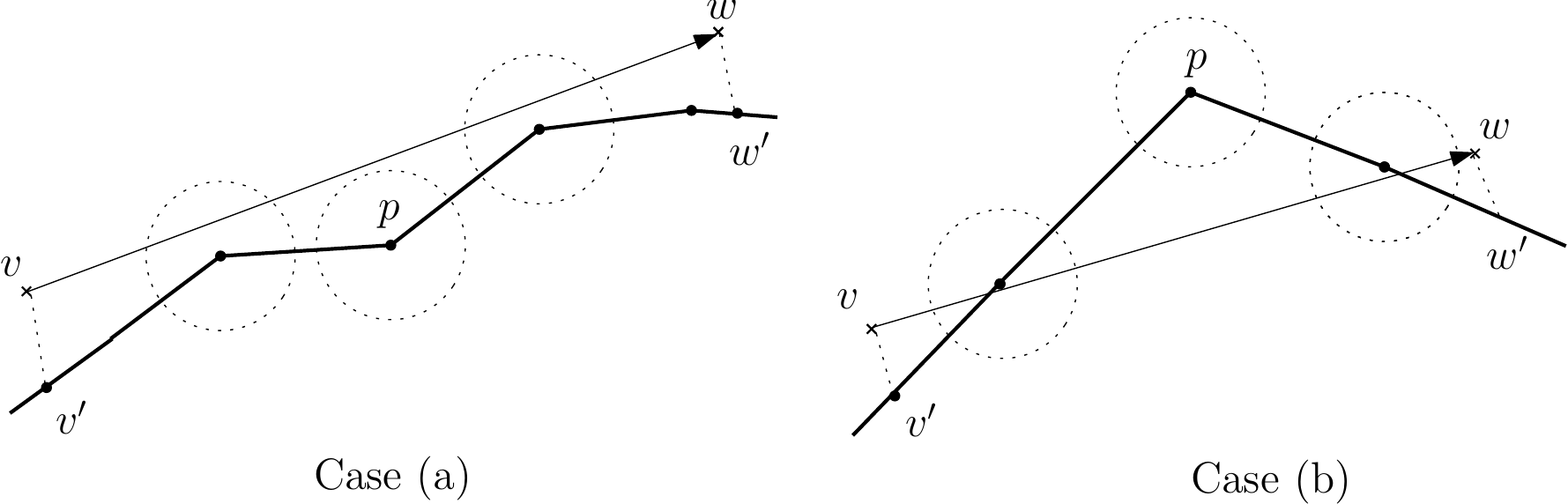}
	\caption{Proof of Lemma \ref{lema:convexhull}}
	\label{fig:convexHull}
\end{figure}

\end{document}